\documentclass[12pt]{iopart}

\usepackage{amssymb}
\usepackage{amsmath} 	
\usepackage{amsfonts,amsthm} 
\usepackage{amsthm} 
\usepackage{mathtools}
\usepackage{graphicx}

\usepackage{amsfonts,amsthm}
\usepackage{bm,bbm}
\usepackage{physics}
\usepackage{mathtools}
\usepackage{xcolor}
\usepackage{hyperref}
\usepackage{tensor}
\usepackage{subcaption} 
\usepackage{soul}
\usepackage{booktabs}
\usepackage{qcircuit}
\usepackage{tikz}

\usepackage{cite}

\def\fC{\mathbb{C}}
\def\fR{\mathbb{R}}

\def\ra{\rangle}
\def\la{\langle}

\def\id{\mathbb{I}}
\def\H{\mathcal{H}}
\def\idA{\vec{\mathbb{I}}_A{} }
\def\idD{\mathbb{I}_D{} }

\def\ra{\rangle}
\def\la{\langle}

\newtheorem{thm}{Theorem}[section]

\newtheorem{lem}[thm]{Lemma}
\newtheorem{cor}[thm]{Corollary}

\theoremstyle{definition}
\newtheorem{defn}[thm]{Definition}

\newtheorem{rmk}[thm]{Remark}

\definecolor{darkgreen}{rgb}{0,.5,0}

\newcommand{\bigzero}{\mbox{\normalfont\Large\bfseries 0}}
\newcommand{\rvline}{\hspace*{-\arraycolsep}\vline\hspace*{-\arraycolsep}}

\eqnobysec

\newcounter{mnotecount}[section]

\begin{document}

\title[]{\centering{Tristochastic operations \\
and convolutions of quantum states}}

\author{Rafał Bistro\'{n}$^{1,2}$, Wojciech \'Smia\l ek$^{3,4}$ and Karol {\.Z}yczkowski$^{1,3}$}

\address{$^1$ Faculty of Physics, Astronomy and Applied Computer Science, Institute of Theoretical Physics, Jagiellonian University, ul. {\L}ojasiewicza 11, 30--348 Krak\'ow, Poland}
\address{$^2$ Doctoral School of Exact and Natural Sciences, Jagiellonian University}
\address{$^3$ Center for Theoretical Physics, Polish Academy of Sciences, Al. Lotnik\'{o}w 32/46, 02-668 Warszawa, Poland}
\address{$^4$ Faculty of Physics, University of Warsaw, ul. Pasteura 5, 02-093 Warszawa, Poland}
\eads{\mailto{rafal.bistron@uj.edu.pl},  \mailto{w.smialek@student.uw.edu.pl},  \mailto{karol.zyczkowski@uj.edu.pl}}
\vspace{10pt}
\begin{indented}
\item[]May 2023
\end{indented}

\begin{abstract}
The notion of convolution of two probability vectors, corresponding to a 
coincidence experiment can be extended for a family of binary operations determined by (tri)stochastic tensors, to describe Markov chains of a higher order. 
 
The problem of associativity, commutativity and the existence of neutral elements and inverses is analyzed for such operations.
For a more general setup of multi-stochastic tensors, we present the characterization of their probability eigenvectors. Similar results are obtained for the quantum case: we analyze tristochastic channels, which induce binary operations defined in the space of quantum states. 
 
Studying coherifications of tristochastic tensors we propose a quantum analogue of the convolution of probability vectors defined for two arbitrary density matrices of the same size. Possible applications of this notion to construct schemes of error mitigation or building blocks in quantum convolutional neural networks are discussed.
\end{abstract}


\maketitle

\section{Introduction}
\label{sec:intro}

Stochastic maps acting on the space of probability vectors are the main tools to describe discrete dynamics in a set of classical probabilistic states. A special, and most widely studied case of such evolution is called a Markov chain and (in the finite scenario) can be characterized by a stochastic matrix $M$, satisfying $M_{ij} \geq 0$ $\sum_i M_{ij} = 1$, for any value of $j$.

An important subcase are bistochastic matrices, also called doubly stochastic, for which the sums of entries in each row and each column is equal to one. 
For this class, many interesting results have been obtained,
 most notably the Birkhoff-von Neumann theorem \cite{Birkhoff_thm}.
Bistochastic matrices, which enjoy interesting spectral properties \cite{bistoch_spectrum}, have been applied 
to various mathematical problems \cite{bistochastic_app, bistrochastic_sperctum}.

Further generalization of these concepts leads us to the idea of cubic stochastic tensors, corresponding to bilinear (or multilinear) stochastic maps in the space of probability vectors. Such construction naturally emerges primarily in the higher-order Markov chains \cite{Limit_points_paper1}.
Any such bilinear map can also be viewed as a product of two probability vectors written $\vec{p} \star_A \vec{q}$ 
 \begin{equation}
\label{classical_convolution_intro}
     (p \star_A q)_i = \sum_{j,k} A_{ijk} p_j q_k~.
 \end{equation}
induced by a stochastic cubic tensor $A$ satisfying $A_{ijk}\geq 0$ and $\sum_i A_{i,j,k} = 1$.

This class of classical channels is very general, hence many important problems for stochastic tensors are unsolvable, or solvable only numerically. The noteworthy example is a problem of finding the generalized eigenvectors of $A$, which resulted in numerous numerical algorithms for finding them \cite{ Limit_points_paper1, Li_14_stationarity, Hu_13_stationarity, Liu_19_stationarity, Yu_20_stationarity}.

We may, however, introduce additional requirements for bistochasticity or tristochasticity of a tensor $A$, requiring the sum of its elements in any of three possible indices to be equal to one. More specific class consists of permutation tensors, which are tristochastic tensors with entries being either $0$ or $1$.
The simplest example of bilinear operation \eqref{classical_convolution_intro} with $A$ being a permutation tensor is the standard convolution between two probability vectors.

Operations induced by these tensors with a maximal degree of stochasticity will be more closely examined in this work.
Naturally arising questions concern the  structure of the set of tristochastic tensors and their generalized eigenvectors.

The main goal of the paper is to analyze quantum analogues of tristochastic tensors and to introduce a convolution of quantum states. 
As it turns out, such an analogy will allow us to translate the results obtained for probability vectors into the framework of density matrices.

As it was noted by Lomont \cite{Lomont}, there is no proper ”operation of convolution”, which for arbitrary two pure states produces a pure state as an outcome.
However, we provide here a construction of quantum convolution of arbitrary density matrices of the same size, which results in a legitimate density matrix. Although some constructions of quantum convolution of density matrices \cite{Aniello_2019} and quantum channels \cite{Sohail_2022} have already been proposed, they usually correspond to a small subset of operations, compared to those that might find useful applications.
Hence in this work we propose, as an analogue of the classical map \eqref{classical_convolution_intro}, 
an original definition of convolution for density matrices as a quantum channel, $\Omega_N \otimes \Omega_N \to \Omega_N$, where $\Omega_N$ denotes the set of density matrices of size $N$. 
Such a definition gives a direct analogue between tristochastic tensors and the proposed quantum tristochastic channels, which allows us to translate several properties of the classical convolution into the quantum domain, such as generalized eigenvectors or neutral elements of the convolution.

Following this analogy, the most natural idea for applying quantum convolution would be to interpret it as a quantum analogue of a coincidence experiment, by the same token as the standard convolution of two probability vectors. 

To provide further motivation for this work, consider an important application of convolution of classical probability vectors: convolutional neural networks 
\cite{Neocognitron}. A quantum version of convolutional neural networks has been recently proposed and proven to be useful in classification tasks, such as determining the phase in spin chains phase transitions, see \cite{QCNN1, QCNN2, QCNN3, Hartmann} and references therein. From this perspective, our work can be understood as an attempt to study the building blocks of quantum convolutional neural networks.
We would like to note here, that the structure of the proposed convolution corresponds to the joint action of convolution and pooling in these networks. 

Another idea \cite{QCNN1} would be to use an "inverse" of quantum convolution to encode some quantum state from $\H$ into greater Hilbert space to make it more robust against correlated noise, and then to decode information to obtain the original state.

The structure of this paper is as follows. In Section \ref{sec:classical}  we analyze the notion of classical tristochastic tensors and classical convolutions. We characterize and describe the generalized eigenvectors \cite{Limit_points_paper1} of  a multi-stochastic tensor, the neutral elements of the convolution and its possible inverses. Analogous results for multi-stochastic quantum channels are established in Section \ref{sec:quantum} and a broader class of quantum convolutions based on coherifications of tristochastic tensors are discussed therein.
In Section \ref{sec:coch} an explicit construction for such a coherification for permutation tensors is presented. Finally, in Section \ref{sec:Qconvolution} we discuss the construction of quantum convolution using quantum circuits and investigate the properties of quantum convolution of arbitrary two states of a single-qubit system.

Appendix \ref{app:on_chanels} contains a note on quantum channels of the form $\mathcal{H} \otimes \mathcal{H} \rightarrow \mathcal{H}$, whereas the Appendix \ref{app:proofs} includes some calculations and technical proofs, for the sake of clarity not provided in the main body of the paper.
Appendix \ref{app:explicit_form} serves as a short recipe, explaining how to construct a quantum convolution.

\newpage
\section{Convolutions of classical probability vectors}\label{sec:classical}

Consider a set $\Delta_N$ of N-dimensional probability vectors:
\begin{equation}
    \Delta_{N} := \left\{ \vec{p}\in \fR^N \ : \ \forall_i \ p_i \geq 0, \ \sum_i p_i = 1 \right\}~,
    \label{prob_simplex}
\end{equation}
which characterize any discrete random variable.
The set of such vectors forms an $N-1$ dimensional simplex $\Delta_N \in \fR^N$, whose vertices are standard unit vectors in $\fR^N$.
Linear transformations of probability vectors are determined by a stochastic matrix $M$,
\begin{equation}
p'_i = \sum_{j = 1}^N M_{ij} p_j~,
\end{equation}
with nonnegative entries, and a fixed sum of entries in each column,
\begin{equation}
    M_{ij} \geq 0, \quad 
    \sum_{j = 1}^N M_{ij} = 1~.
    \label{stoch_matrix}
\end{equation}
These matrices of order $N$ describe discrete dynamics in the set $\Delta_N$ of classical probabilistic states, that depend only on the current state of the system. Such processes are famously known as \textit{Markov chains}.

In the space of stochastic matrices satisfying \eqref{stoch_matrix},
one distinguishes the subset of bistochastic matrices, which additionally fulfil the dual condition for the sum in each row,
\begin{equation}
    B_{ij} \geq 0, \quad \sum_{i = 1}^N B_{ij} = \sum_{j = 1}^N B_{ij} = 1 ~.
\end{equation}
Bistochastic matrices have several interesting mathematical properties \cite{Birkhoff_thm, bistoch_spectrum}, and they are often applied to problems in physics and the theory of communication \cite{bistochastic_app}. 

A natural next step is to generalize these concepts and introduce a class of bilinear operations on probability vectors, determined by  cubic stochastic tensor $A$, 

\begin{equation}
\label{vec_conv}
    r_i = (p\star_A q)_i = \sum_{j,k = 1}^N A_{i jk} p_j q_k~,
\end{equation}
which describes the evolution of a second order Markov chain \cite{Limit_points_paper1}.
Moreover, any such operation describes a certain generalized coincidence experiment. In particular, this class includes the convolution and correlation of probability vectors.

The counterpart to bistochasticity for cubic tensor can be defined as:
\begin{defn}
    Rank three tensor $A$ is called \textit{tristochastic} if all of its entries are nonnegative and satisfy the following three conditions:
    \begin{equation}
    \label{tristoch_prob_eq}
        \sum_{i = 1}^N A_{i j k} = \sum_{j = 1}^N A_{i j k} = \sum_{k = 1}^N A_{i j k} = 1~.
        \end{equation}
\end{defn}
A special subclass of tristochastic tensors is formed by permutation tensors.
\begin{defn}
    \textit{Permutation tensor} $A$ is a tristochastic tensor with all of the entries equal to either 0 or 1.
\end{defn}

An example of a tristochastic permutation tensor for $N = 3$ is provided below.

\begin{equation}
\label{example_permutation_tensor}
T_3 = (\id, P_3, P_3^2)
= \left(\begin{matrix}
1 & 0 & 0\\
0 & 1 & 0\\
0 & 0 & 1\\
\end{matrix}\right.
\left|\begin{matrix}
0 & 1 & 0\\
0 & 0 & 1\\
1 & 0 & 0\\
\end{matrix}\right.
\left|\begin{matrix}
0 & 0 & 1\\
1 & 0 & 0\\
0 & 1 & 0\\
\end{matrix}\right)~,
\end{equation}
where $P_n$ denotes the cyclic matrix of order $N$.
It is easy to see that such a tensor consists of $N^2$ entries equal to $1$ and $N^2(N - 1)$ entries equal to $0$.
Tristochastic tensors lead to a generalization of the convolution of discrete probability vectors.
It is worth investigating, which properties of bistochastic matrices generalize to higher-rank tensors with a maximal degree of stochasticity. To make our arguments  general, we shall use the following notion for the tensors of rank $m$.

\begin{defn}
A tensor $A_{i_1,\cdots,i_m}$ is called a $m-$stochastic tensor, if all entries of $A$ are nonnegative, and for any index $i_k$ one has $\sum_{i_k} A_{i_1, \cdots,i_m} = 1$
\end{defn}

To not confuse the operation defined by $m-$stochastic tensor with multiple actions of tristochastic channel, we denote the former by $A[\vec{p}_1, \vec{p}_2,\cdots]$.

\subsection{Associavity and commutativity of stochastic product}

For a generic stochastic tensor, the induced product will be neither commutative nor associative. For certain classes of stochastic tensors, however, the product may acquire these important properties. Below we present the necessary and sufficient conditions for these properties.

\begin{thm}\label{commute}
Let $A$ be an N-dimensional cubic tensor satisfying conditions
(\ref{tristoch_prob_eq}) defining a stochastic map
$(\vec{p} \star_A \vec{q})_i = \sum_{j,k} A_{ijk} p_j q_k$, where $p,q$ are N-dimensional probability vectors. 
The operation $\star_A$ is commutative if and only if the tensor $A_{ijk}$ is symmetric with respect to the exchange of the last two indices: $\forall_i A_{ijk} = A_{ikj}$. Moreover, $\star_A$ is associative if and only if $\sum_i A_{ail} A_{ijk} = \sum_i A_{aji} A_{ikl}$.
\end{thm}

\begin{proof}
Condition of commutativity of the product $\star_A$ is equivalent to the following,
\begin{equation}
\sum_{j,k=1}^N A_{ijk} p_j q_k = \sum_{j,k=1}^N A_{ijk} q_j p_k = \sum_{j,k=1}^N A_{ikj} p_j q_k~.
\end{equation}
As the condition needs to be satisfied for any input vectors, the equality holds if and only if $A_{ijk} = A_{ikj}$. By the same token one may prove the associativity.
\end{proof}

The simplest example of a tristochastic tensor inducing an associative convolution reads $A_{ijk} = r_{i - j - k}$, where $\vec{r}$ is a probability vector. More involved example is given by a "classical analogue" of the convolution proposed in \cite{Aniello_2019},

\begin{equation}
\label{r_conv_classical}
    \vec{p} \star_{\vec{r}} \vec{q} := \frac{1}{(N -1)!} \sum_{\sigma \in \Sigma(N)} (\vec{p}\cdot P_{\sigma} \vec{r}) P_{\sigma}^{-1} \vec{q}~,
\end{equation}
where $\Sigma(N)$ is a set of all permutations of elements from $\{1,\cdots,N\}$ and $P_{\sigma}$ is a matrix representing a permutation $\sigma$. The associativity is straightforward to prove by direct calculations. To establish the tristochasticity of \eqref{r_conv_classical} it is sufficient to notice that conditions \eqref{tristoch_prob_eq} for tristochasticity  are equivalent to the following statement: For any probability vector $p$ one has

\begin{equation}
\label{alternative_tristochasticity1}
\vec{e} \star_{\vec{r}} \vec{p} = \vec{p} \star_{\vec{r}} \vec{e} = \vec{e}~,
\end{equation}
see Lemma \ref{lem_mstoch1}.
Here $\vec{e} = (\frac{1}{N},\frac{1}{N},\cdots )$ denotes the maximally mixed probability vector. These conditions are  easy to verify,

\begin{equation}
\label{alternative_tristochasticity_proof}
\begin{aligned}
& \vec{e} \star_{\vec{r}} \vec{p} =  \frac{1}{(N-1)!} \sum_{\sigma \in \Sigma(N)} (\vec{e}\cdot P_{\sigma} \vec{r})~ P_{\sigma}^{-1} \vec{q} = \frac{1}{(N-1)!} \sum_{\sigma \in \Sigma(N)} \frac{1}{N} P_{\sigma}^{-1} \vec{q} = \frac{1}{N!} \sum_{\sigma \in \Sigma(N)} P_{\sigma}^{-1} \vec{q} = \vec{e}~, \\
& \vec{p} \star_{\vec{r}} \vec{e} =  \frac{1}{(N-1)!} \sum_{\sigma \in \Sigma(N)} (\vec{p}\cdot P_{\sigma} \vec{r})~ P_{\sigma}^{-1} \vec{e} = \vec{e} ~~\frac{1}{(N-1)!} \sum_{\sigma \in \Sigma(N)} (\vec{p}\cdot P_{\sigma} \vec{r}) = \vec{e}~. \\
\end{aligned}    
\end{equation}

\subsection{Fixed point of convolution and eigenvectors of multi-stochastic tensors}

In this subsection, we resolve the issue of generalized eigenvectors for multi-stochastic tensors. Firstly, we focus on the map $p \to A[p,p,\cdots]$, since all of its fixed points are by definition generalized eigenvectors of $A$. Then using the reducibility of multi-stochastic tensors we present an explicit description of all probability eigenvectors of $A$ by its tristochastic subtensors.

Let us start by generalizing the invariance property of bistochastic matrices.

\begin{lem}\label{lem_mstoch1}
Let $A$ be an arbitrary $m$-stochastic channel. Then for any sequence of probability vectors $\{\vec{p}_2, \cdots ,\vec{p_m}\}$ one of which is $\vec{e}$, the following equality holds:
\begin{equation}
\label{multi-stochastic_e}
A[\vec{p}_2, \cdots, \vec{e}, \cdots, \vec{p}_m ] = \vec{e}
\end{equation}
\end{lem}

\begin{proof}
Without loss of generality let $\vec{e}$ be the $k$-th vector in the sequence $\{\vec{p}_1, \cdots ,\vec{p_m}\}$. The left hand side of eq. \eqref{multi-stochastic_e} can be written as,
\begin{equation*}
\begin{aligned}
& A[\vec{p}_2, \cdots, \vec{e}, \cdots, \vec{p}_m ]_{i_1} = \sum_{i_2, \cdots, i_m} A_{i_1 \; \cdots \; i_k, \cdots\; i_m} {p_2}_{\;i_2}, \cdots, \vec{e}_{i_k}, \cdots, {p_m}_{\;i_m} = \\
& = \sum_{i_2, \cdots, i_{k-1}, i_{k+1}, \cdots, i_m} \left(\sum_{i_k} A_{i_1\; \cdots\; i_k\; \cdots\; i_m} \frac{1}{N}\right) {p_2}_{\;i_2}, \cdots, {p_{k-1}}_{\;i_{k-1}}, {p_{k+1}}_{\;i_{k+1}}, \cdots, {p_m}_{\;i_m} =  \\
& = \sum_{i_2, \cdots, i_{k-1}, i_{k+1}, \cdots, i_m} \frac{1}{N} {p_2}_{\;i_2}, \cdots, {p_{k-1}}_{\;i_{k-1}}, {p_{k+1}}_{\;i_{k+1}}, \cdots, {p_m}_{\;i_m} = \frac{1}{N}
\end{aligned}
\end{equation*}
\end{proof}

Using this fact, we describe the fixed point of the map, $\vec{q} \to A[\vec{q},\cdots,\vec{q}]$, for any $\vec{q}$ in the interior of the probability simplex.

\begin{thm}\label{multisto_fixed}
Let $\vec{q}^{(0)}$ be any point in the interior of probability simplex $\Delta_N$ and $\vec{p}^{(0)}$ be a point at its boundary  $\partial\Delta_N$, such that $\vec{q}^{(0)} = \alpha \vec{e} + (1 -\alpha)\vec{p}^{(0)}$ for some $\alpha \in (0,1]$. For any $m$-stochastic map $A$ let us denote the sequences $\{\vec{q}^{(n)}\}$ and $\{\vec{p}^{(n)}\}$ by $\vec{q}^{(n+1)} = A[\vec{q}^{(n)},\cdots, \vec{q}^{(n)}]$ and $\vec{p}^{(n+1)} = A[\vec{p}^{(n)},\cdots, \vec{p}^{(n)}]$. Then if $m > 2$ the sequence $\{\vec{q}^{(n)}\}$ converges to $\vec{e}$.
\end{thm}

Lengthy proof of this theorem is relegated to Appendix \ref{app:clasical_conv}
For further work we need one more notion.

\begin{defn} \cite{Limit_points_paper1}
An $m$-stochastic tensor  $A_{i_1, \cdots,i_m}$ is called reducible, if there exist a proper subset $I \subset \{1,\cdots, N\}$, such that
\begin{equation}\label{reducible_def}
    \forall~ i_1 \in I~,~ i_2,\cdots, i_m \notin I ~~A_{i_1\;\cdots\;i_m} = 0 ~.
\end{equation}
An $m$-stochastic tensor is called irreducible if it is not reducible.
\end{defn}

\begin{lem}\label{sub_multi}
Let $A_{i_1\; \cdots\; i_m}$ be an $m$-stochastic reducible tensor with $m \geq 3$. Moreover, let $I$ be a set of indices defined as above with $\#I = k$. Then the following statements hold
\begin{enumerate}
    \item $k \geq N/2$~,
    \item Tensor $A_{i_1\; \cdots\; i_m}$ is reducible with respect to any of its indices, with the same set $I$ of index values, 
    \begin{equation}
    \forall r \in \{1\; \cdots\; m\}~,~ \forall~ i_r \in I~,~ i_1,\cdots, i_m \notin I  ~~, A_{i_1\;\cdots\;i_r\;\cdots,i_m} = 0~,
    \end{equation}
    \item The truncation of the tensor $A_{i_1\;\cdots\;i_m}$:
    \begin{equation}\label{sub_stochastic_1}
        A'_{i_1\;\cdots\; i_m} = A_{i_1\;\cdots\; i_m} ~,~ i_1, \cdots, i_m \notin I
    \end{equation}
    is also an $m$-stochastic tensor. Moreover, for any $m$-stochastic tensor  $A_{i_1\;\cdots\; i_m}$, if there exist a subset $I \subset \{1,\cdots,N\}$ such that a tensor in  \eqref{sub_stochastic_1} is also an $m$-stochastic tensor, then $A_{i_1\;\cdots\;i_m}$ is reducible.
\end{enumerate} 
\end{lem}

Proof of this lemma is provided in Appendix \ref{app:clasical_conv}. We are now ready to present the main theorem of this section.

\begin{figure}[h]
    \centering
        \includegraphics[height=2.7in]{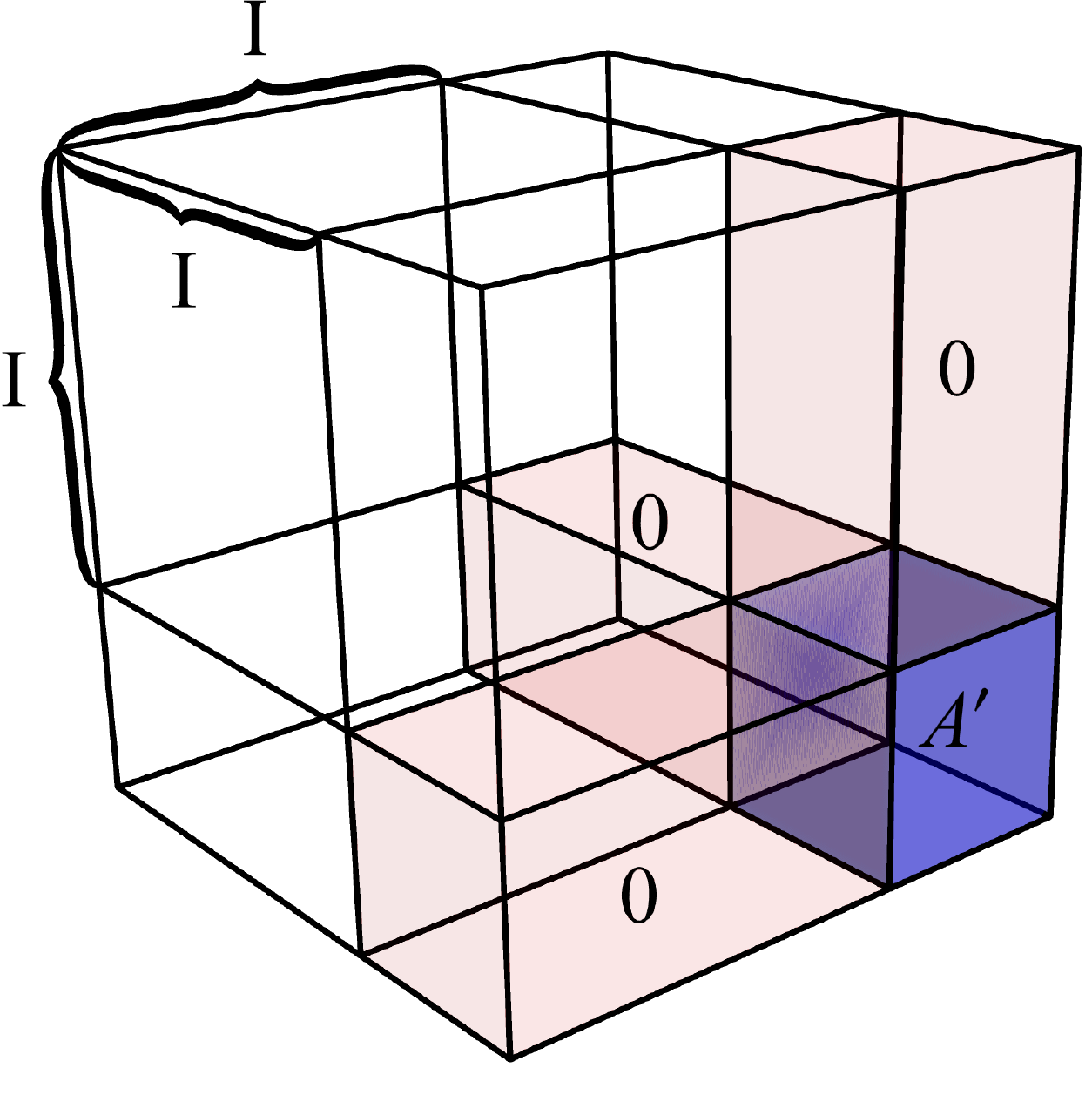}
     \centering
     \caption{\label{fig:Riem}Schematic representation of a reducible tensor $A$ with marked set of indices $I$ and a tristochastic subtensor $A'$. }
\end{figure}

\begin{thm}\label{m_stoch_tensor_finale}
Let $A$ be an $m$-stochastic tensor and $m \geq 3$. Then each subset of indexes values $I \subset \{1,\cdots N\}$ such that
\begin{equation*}
    A_{i_1\;\cdots\;i_m} = 0~~ \forall~ i_1 \in I~,~ i_2,\cdots, i_m \notin I  ~,
\end{equation*}
corresponds to exactly one eigenvector of $A_{i_1,\cdots, i_m}$ inside the probability simplex $\Delta_N$ and vice versa. Moreover, such an eigenvector has the form
\begin{equation}\label{secound_gen_eing}
    p_i = \left\{ \begin{matrix}
0 ~\text{ if }~ i \in I \\
\frac{1}{N-k} ~\text{ if }~ i \notin I
\end{matrix}   \right.
\end{equation}
where $k = \#I$. Each such generalized eigenvector corresponds to the eigenvalue $1$.
\end{thm}

\begin{proof}
Firstly we prove that each eigenvector $\vec{p}$ corresponds to only one set of indexes $I$ described in the Theorem and then show that each set of indexes $I$ corresponds to a single eigenvector $\vec{p}$. Along the way, we prove a second statement of the Theorem. 

Let $\vec{p}$ be an eigenvector of $A$corresponding to the eigenvalue $\lambda$. If all $p_i$ are greater than zero, then by Theorem \ref{multisto_fixed} one has $\vec{p} = \vec{e}$ and $\lambda = 1$, since $\vec{p}$ belongs to the interior of the probability simplex. Hence the theorem is satisfied with $I = \emptyset$.
Otherwise, there  exist a nonempty subset $I \subset \{1,\cdots,N\}$  such that $p_{i} = 0$ for all $i \in I$. Then for each $i \in I$ one has:
\begin{equation}
\begin{aligned}
0 & = \lambda p_{i}=  A[\vec{p},\cdots,\vec{p}]_{i}  = \sum_{i_2,\cdots,i_m} A_{i\;i_2\;\cdots\; i_m} p_{i_2}\cdots p_{i_m} = \sum_{i_2,\cdots, i_m\notin I } A_{i\; i_2,\cdots\; i_m} p_{i_2}\cdots p_{i_m}~.
\end{aligned}
\end{equation}
For each combination of indices $i_2,\cdots, i_m \notin I$ the product $p_{i_2}\cdots p_{i_N}$ is positive. Thus  $A_{i_1\;\cdots\; i_m} = 0$ for each $i_1 \in I$ and each $i_2,\cdots, i_N \notin I$, which proves the claim in one direction.

Consider a subspace defined by indices not belonging to the set $I$. Then by Lemma \ref{sub_multi}, the tensor $A_{i_1\;\cdots\;i_m}$ can be truncated to an $m$-stochastic tensor,
\begin{equation}\label{A_truncated}
A'_{i_1\;\cdots\; i_m} = A_{i_1\;\cdots\; i_m} ~,~ i_1, \cdots, i_m \notin I~,
\end{equation}
and truncated $\vec{p}$ is its eigenvector with all positive entries. Using  Theorem \ref{multisto_fixed} to  tensor \eqref{A_truncated}, the vector $\vec{p}$ must be maximally mixed on the discussed subspace. Hence it has the structure \eqref{secound_gen_eing} and forms a generalized eigenvector of $A$ with an eigenvalue equal to unity.

To prove the implication in the opposite direction we use Lemma \ref{sub_multi}. Because a tensor $A'_{i_1\;\cdots\;i_m} =  A_{i_1\;\cdots\; i_m} ~,~ i_1,\cdots, i_m \notin I$ is $m$-stochastic, by Lemma \ref{lem_mstoch1} it also has a generalized eigenvector $\vec{p} := \vec{e}$ on the $N-k$ dimensional subspace described by the values of indices not in $I$. Hence $\vec{p}$ has a form \eqref{secound_gen_eing}.
Moreover, since $A_{i_1\;\cdots\;i_m}$ is reducible, the second point of Lemma \ref{sub_multi} holds, so $\vec{p}$ is also an eigenvector of $A_{i_1\;\cdots\;i_m}$.
\end{proof}

\subsection{Identity and inverses}

Let us review the convolution of two probability vectors as a binary operation and focus on the identity and inverse of convolution.

\begin{defn}\label{id_def}
Let $A$ be a tristochastic tensor and $\idA$ a probability vector such that for any probability vector $\vec{p}$ one has
\begin{equation*}
    \idA \star_A\vec{p} = \vec{p}\star_A\idA  = \vec{p}~.
\end{equation*}
Then $\idA$ is called an \textit{identity} of the product $\star_A$ generated by the tensor $A$.
\end{defn}

\begin{thm}\label{classical_identity_thm}
Let $A$ be a tristochastic tensor and $\idA$ an identity of $\star_A$ then: 
\begin{enumerate}
    \item $\idA$ is a nontrivial eigenvector of $A$, hence $A$ is reducible,
    \item $\idA$ is of the form $(0,\cdots,0,1,0,\cdots,0)$, i.e. it has only one nonzero value on $k$-th place,
    \item For any $i,j \in \{1,\cdots,N\}$, and $k$ as in the previous point $A_{ijk} = A_{ikj} = \delta_{i,j}$,
    \item Moreover, if all of the above points are true for some vector $\vec{q}$, then $\vec{q}$ is an identity of $A$.
\end{enumerate}
\end{thm}

\begin{proof}
By the definition of identity $\idA \star_A \idA = \idA$, which implies the statement $(i)$.

To prove the point $(ii)$ let $I$ be a set of indices such that $(\idA)_i = 0$. Since 
$\idA$ is a maximally mixed state on subspace corresponding to $\{1,\cdots,N\}\setminus I$, then for any probability vector
$\vec{q}$ with nonzero values only for indices
form $\{1,\cdots,N\}\setminus I$,
\begin{equation*}
    \idA \star_A\vec{q}  = \idA~.
\end{equation*}
By the definition of identity  $\idA \star_A \vec{q} = \vec{q}$, therefore $q = \idA$, so $\{1,\cdots,N\}\setminus I$ is a one element set.

The point $(iii)$ follows form the calculations performed below for any probability vector $\vec{p}$,
\begin{equation*}
\begin{aligned}
& p_i = \sum_{j,k} A_{ijk} p_j \idA_k = \sum_j A_{ijk} p_j~, \\
& p_i = \sum_{j,k} A_{ikj} \idA_k p_j  = \sum_j A_{ikj} p_j ~.\\
\end{aligned}
\end{equation*}

To prove the last statement $(iv)$ one calculates,
\begin{equation*}
\begin{aligned}
& (\vec{q} \star_A\vec{p})_i = \sum_{j,k'} A_{ikj} q_{k'} p_j = \sum_{j} A_{ikj} p_j = \sum_j \delta_{i,j} p_j = p_j~, \\
& (\vec{p} \star_A \vec{q})_i = \sum_{j,k'} A_{ijk'} p_j q_{k'} = \sum_{j} A_{ijk} p_j = \sum_j \delta_{i,j} p_j = p_j~, \\
\end{aligned}
\end{equation*}
where the second equality follows from $\vec{q}$ satisfying point $(ii)$, and the third equality follows from $\vec{q}$ satisfying point $(iii)$.
\end{proof}

\begin{thm}\label{classical_inv_thm}
Let $A$ be a tristochastic tensor and $\idA$  be an identity of $\star_A$. Then for a probability vector $\vec{p}$ there exist an inverse probability vector $\vec{q}$,
\begin{equation}
\vec{p} \star_A \vec{q} = \vec{q} \star_A \vec{p}  = \idA~,
\end{equation}
if and only if $\vec{p}$ is of the form $p_i = \delta_{i,m}$ for some $m \in\{1,\cdots,N\}$ and there exist a $n \in\{1,\cdots,N\}$ such that $A_{kmn} = A_{knm} = 1$. 
\end{thm}

\begin{proof}
Assume there exists $\vec{q}$, which is an inverse of $\vec{p}$.
Since the identity is of the form $\idA_i = \delta_{k\;i}$, then for each $i \neq k$,
\begin{equation*}
0 = \idA_i = \sum_{jl} A_{ijl} p_j q_l~.
\end{equation*}
Because each term in the above sum is nonnegative, it implies that if $p_j > 0$ and $q_l> 0$, then $A_{ijl} = 0$ for any $i \neq k$.
On the other hand $\sum_i A_{ijl} = 1$, hence for $j$ and $l$ such that $p_j > 0$ and $q_l> 0$, one has  $A_{kjl} = 1$. Let us pick any $l$ such that $q_l>0$, then:
\begin{equation*}
    1 = \sum_{j} A_{kjl} \geq \sum_{j ~:~ p_j > 0} A_{kjl}  = \sum_{j ~:~ p_j > 0} 1~.
\end{equation*}
Hence $\vec{p}$ can have at most one nonzero element, $p_i = \delta_{i,m}$. Analogically, we can show that $q_l = \delta_{l,n}$.
Finally, because $\idA = \vec{p} \star_A \vec{q} = \vec{q} \star_A \vec{p} $, one has,
\begin{equation*}
\begin{aligned}
1 = \idA_k = \sum_{j,l}A_{kjl}p_jq_l = A_{kmn} \text{~ and ~} 1 = \idA_k = \sum_{j,l}A_{klj}q_lp_j = A_{knm}~.
\end{aligned}
\end{equation*}
The implication in the opposite direction is trivial.
\end{proof}

From Theorem \ref{classical_identity_thm} one immediately sees that for generic tristochastic tensor, the identity of convolution usually does not exist. However, in each dimension, it is quite easy to construct tristochastic tensors, for which the convolution always possesses an identity, for example
\begin{equation*}
    A_{kij} = (P_N)_{ij}^{k-1}~,
\end{equation*}
where $P_N$ is a permutation matrix corresponding to a cyclic permutation $\sigma(i) = i+1$ mod $N$. In the dimension $3$ the permutation tensor of interest is provided in \eqref{example_permutation_tensor}.
Moreover, Theorem \ref{classical_identity_thm} implies another property of the identity of convolution.

\begin{cor}
For each tristochastic tensor $A$ the convolution $\star_A$ possesses at most a single identity.
\end{cor}

\begin{proof}
To prove this statement by contradiction assume that the product
$\star_A$ possesses two identities. Then by Theorem  \ref{classical_identity_thm}, item  $(iii)$ for each pair of indices $i,j$ and some $k_1$, $k_2$ we have
\begin{equation*}
A_{ijk_1} = A_{ijk_2} = \delta_{ij}~.
\end{equation*}
Hence evaluating the expression $\sum_{k} A_{i\;i\;k}$, for any value of $i$ we get:
\begin{equation*}
\sum_k A_{iik} \geq A_{iik_1} + A_{iik_2} = 1 + 1  = 2,
\end{equation*}
which contradicts the tristochasticity of $A$.
\end{proof}

\section{Quantum multi-stochastic operations}\label{sec:quantum}

Now we change our focus to a quantum system. Our main constituent is a set $\Omega_N$ of density matrices of size $N$:
\begin{equation*}
\Omega_N = \{ \rho:\;\; \rho \geq 0,\; \Tr[\rho] = 1 \}  
\end{equation*}

The transformations of interest between density matrices are completely positive trace preserving maps, $\Psi: \Omega_N \to \Omega_N$ called quantum channels. Let $|\varphi^+\ra \in \H_N^{(1)}\otimes \H_N^{(2)}$ denote the maximally entangled state $|\varphi^+\ra = \sum_{i} \frac{1}{\sqrt{N}}|i\ra\otimes|i\ra$. The Choi-Jamiołkowski isomorphism \cite{Choi,Jamiolkowski} implies that any quantum channel $\Psi$ can be represented by a Jamiołkowski state $\rho_{\Psi}$ on an extended system as
\begin{equation*}
\rho_{\Psi} = (\Psi \otimes \id) |\varphi^+\ra\la\varphi^+|~.
\end{equation*}
The rescaled state $D = N \rho_{\Psi}$ is called a dynamical matrix of $\Psi$ or a Choi matrix and allows to express the action of a channel $\Psi$ as
\begin{equation}
\label{simple_channel}
\Psi(\rho) = \Tr_2[D (\id \otimes \rho^\top)] ~,
\end{equation}
where $D \geq 0$ and $\Tr_2$ is a partial trace over the second subsystem. Moreover to assure the trace preserving condition the dynamical matrix $D$ satisfies
\begin{equation*}
    \Tr_1[D] = \id~.
\end{equation*}

Quantum operations satisfying the unitality condition $\Psi_B(\id) = \id$ are called bistochastic and form an important class of the channels. The dynamical matrix of any bistochastic channel $\Phi_B$ satisfies two conditions for both partial traces,
\begin{equation}
\Tr_{1}[D] = \id_N, ~~~~~~~~~~ \Tr_{2}[D] = \id_N~.
\end{equation}

Making use of the channel description \eqref{simple_channel}, we propose a quantum convolution inspired by the convolution \eqref{vec_conv} of classical probability vectors.

\begin{defn}\label{quant_conv_def}
A \textit{quantum convolution} is a channel $\Phi: \Omega_N \otimes \Omega_N \to \Omega_N$ defined by 
\begin{equation}
\label{quantum_convolution}
\rho \star_D \sigma := \Phi_D[\rho,\sigma] := \Tr_{23} [D_{123}(\id \otimes \rho^\top \otimes \sigma^\top)]~,
\end{equation}
where a tri-partite dynamical matrix $D_{123}$ is an arbitrary semipositive defined matrix of order $N^3$, which satisfies a single  partial trace condition, $\Tr_1[D] = \id_{N^2}$.
\end{defn}

For our work, it is essential to specify a class of convolutions determined by a tristochastic matrix $D_{123}$, which satisfies three conditions analogous to \eqref{tristoch_prob_eq}, 
\begin{equation}
\Tr_1[D_{123}] = \id_{N^2}  ~~~~~~~~~~  \Tr_2[D_{123}] = \id_{N^2} ~~~~~~~~~~ \Tr_3[D_{123}] = \id_{N^2}~.
\end{equation}

Moreover, one can check that these conditions are equivalent to the statement that channels defined via 
\begin{equation}
\label{quant_tri}
\Tr_{13}[D (\rho^\top \otimes \id \otimes \sigma^\top)] \text{~~ and ~~} \Tr_{12}[D (\rho^\top \otimes \sigma^\top \otimes \id)]  
\end{equation}
are also well defined trace preserving quantum channels. For brevity the indices $1,2,3$, emphasizing that $D_{123}$ is a tripartite state, will be suppressed in the future part of the work, so we shall write $D_{123} = D$.

In analogy to the classical case, let us also define a general $m$-stochastic channel, which we use to make our proofs more general: 

\begin{defn}\label{quant_mstoch_def}
Channel $\Phi_D: \Omega_N^{\otimes (m-1)} \to \Omega_N$ defined by a tensor $D_{j_1, \cdots, j_m}^{i_1, \cdots, i_m}$ with $m$ indices by
\begin{equation}
\Phi_D[\rho_2 \otimes \cdots \otimes \rho_{m-1}] = \Tr_{1, \cdots,m-1}[D (\mathbb{I} \otimes \rho_1^\top \otimes \cdots \otimes \rho_{m-1}^\top)]~,
\end{equation}
is called an $m$-stochastic channel, if for any sequence of density matrices $\{\rho_1, \cdots, \rho_{m-1}\}$ and any index $k \in \{1, \cdots, m\}$ the map:
\begin{equation}\label{eq_multi_D2}
\Tr_{1,\cdots,k-1,k+1,\cdots,m-1 }\;[D(\rho_1^\top \otimes \cdots \otimes \underset{k\text{-th element}}{\mathbb{I} } \otimes \cdots \otimes \rho_{m-1}^\top) ]~,
\end{equation}
also forms a valid quantum channel.
\end{defn}

The question, of whether a quantum convolution is commutative or associative depends on the choice of the dynamical matrix $D$  and is analogous to the classical problem. It is worth mentioning that the convolution proposed by Aniello \cite{Aniello_2019}, which for the Lie group $SU(N)$ has a form
\begin{equation}
    \rho \star_{\vartheta} \sigma = \int_{SU(N)} d\mu(U) \Tr[\rho U \vartheta U^\dagger] U \sigma U^\dagger,
\end{equation}
where $\rho$, $\sigma$ and $\vartheta$ are density matrices,
is a tristochastic associative convolution. The proof of its tristochasticty is analogous to the proof in the classical case \eqref{alternative_tristochasticity_proof}.

The last important definition needed in the quantum framework is reducibility.

\begin{defn}\label{channel_reducible_1}
An $m$-stochastic channel  $\Phi_D$ is called reducible if there exist a proper subspace $V\in \H$, such that
\begin{equation}\label{D_reducible_def}
    \forall~ |\phi_1\ra \in V~,~ |\psi_2\ra,\cdots, |\phi_m\ra \in V^\perp ~~\Tr[D(|\phi_1\ra\la\phi_1| \otimes |\phi_2\ra\la\phi_2| \otimes\cdots\otimes |\phi_m\ra\la\phi_m|)  ] = 0 ~.
\end{equation}
An $m$-stochastic channel is called irreducible if it is not reducible.
\end{defn}

With those definitions, we are prepared to present the first main result of this Section, which is a direct analogue of Theorem \ref{m_stoch_tensor_finale}.

\begin{thm}\label{general_eigenvector_quantum}
Let $\Phi_D$ be an $m$-stochastic channel acting on the $N$-dimensional states and $m \geq 3$. Then each subspace $V\subset \H$ such that
\begin{equation*}
\forall~ |\phi_1\ra \in V~,~ |\psi_2\ra,\cdots, |\phi_m\ra \in V^\perp ~~Tr[D(|\phi_1\ra\la\phi_1| \otimes |\phi_2\ra\la\phi_2| \otimes\cdots\otimes |\phi_m\ra\la\phi_m|)  ] = 0 ~.
\end{equation*}
correspond to exactly one density matrix being a generalized eigenvector of $\Phi_D$ and vice versa. Moreover, each such eigenvector is the maximally mixed state on the subspace $V^{\perp}$ and the corresponding eigenvalue is $1$.
\end{thm}

Proof of this theorem, analogous to the proof of the theorem \ref{m_stoch_tensor_finale}, is presented in Appendix \ref{app:quant_fix_point} together with all intermediate steps.

One can also form the quantum analogues of Theorems  \ref{classical_identity_thm} and \ref{classical_inv_thm} established for the classical convolution.

\begin{defn}\label{q_id_def}
Let $\Phi_D$ be a tristochastic channel established in the Definition \ref{quant_mstoch_def} and $\idD$ a density matrix such that for any density matrix $\rho$:
\begin{equation*}
    \idD \star_D \rho = \rho \star_D \idD  = \rho~,
\end{equation*}
then the density matrix $\idD$ is called an identity of $\Phi_D$.
\end{defn}

\begin{thm}\label{quant_identity_thm}
Let $\Phi_D$ be a tristochastic channel and $\idD$ an identity of $\Phi_D$. Then the following statements hold 
\begin{enumerate}
    \item $\idD$ is a nontrivial eigenvector of $\Phi_D$, hence $\Phi_D$ is reducible,
    \item $\idD$ is a pure state, spanning one dimensional subspace $V_{\id}$,
    \item Let $P_{\id}$ be a projection on the subspace defined above, then
    \begin{equation*}
        \Tr_{3}[(\id_N \otimes \id_N \otimes P_{\id}) D (\id_N \otimes \id_N \otimes P_\id)] =  \Tr_{2}[(\id_N \otimes P_\id \otimes \id_N) D (\id_N \otimes P_{\id} \otimes \id_N)] = D_{\id}~,
    \end{equation*}
    where $[D_{\id}]_{\; j\; l}^{i\; k \;} = \delta_{ik}\delta_{jl}$ is a dynamical matrix of the identity channel,
    \item If all of the above points are true for some density matrix $\rho$, then $\rho$ is an identity of $\Phi_D$,
    \item Moreover, for each tristochastic channel $\Phi_D$, there exists at most a single identity state.
\end{enumerate}
\end{thm}

\begin{thm}\label{quant_inv_thm}
Let $\Phi_D$ be a tristochastic channel and $\idD$ an identity of $\Phi_D$. Then for density matrix $\rho$ there exist an inverse density matrix $\sigma$,
\begin{equation}
\rho \star_D \sigma = \sigma \star_D \rho  = \idD~,
\end{equation}
if and only if $\rho$ is a pure state spanning one dimensional subspace $V_{\rho}$ and there exist a one dimensional subspace $V_{\sigma}$ such that
\begin{equation*}
\Tr[(P_{\id} \otimes P_{\rho} \otimes P_{\sigma}) D (P_{\id} \otimes P_{\rho} \otimes P_{\sigma})] = \Tr[(P_{\id} \otimes P_{\sigma} \otimes P_{\rho}) D (P_{\id} \otimes P_{\sigma} \otimes P_{\rho})] = 1~,
\end{equation*}
where $P_\id$ is a projection onto a subspace spanned by an identity of $\Phi_D$, $P_\rho$ is a projection on a subspace $V_{\rho}$ and $P_\sigma$ is a projection on a subspace $V_{\sigma}$.
The inverse of $\rho$, $\sigma$ is a state from $V_{\sigma}$.  
\end{thm}

Theorems \ref{quant_identity_thm} and \ref{quant_inv_thm} can be proven in the analogy to the proofs of their classical counterparts.
One just needs to rephrase all steps in terms of density matrices instead of probability vectors. 

\subsection{Coherification of $m$-stochastic tensors}

Given two diagonal density matrices $\rho$ and $\sigma$ one can calculate their convolution in two ways: as a quantum convolution or as a classical convolution of their diagonal elements:

\begin{equation}
\begin{aligned}
\label{different_diag_conv}
    &\rho \star_D \sigma = \Tr[D(\id \otimes\rho^\top \otimes \sigma^\top)] \\
    &\left(\text{diag}(\rho) \star_A \text{diag}(\sigma)\right)_k = \sum_{ij} A_{k,i,j} \text{diag}(\rho)_i \text{diag}(\sigma)_j
\end{aligned}
\end{equation}

The results of these two operations agree for any diagonal $\rho$ and $\sigma$ if and only if $D_{\; i \; j \; k}^{i \; j \; k} = A_{ijk}$. By this token, we may define a \textit{coherification} \cite{Korzekwa_coherifying} of an $m$- stochastic tensor.

\begin{defn}\label{def_coch}
A coherification of  the $m$-stochastic tensor $A$ is a channel $\Phi_D:\Omega_N^{m-1} \to \Omega_N$, such that the diagonal of its dynamical matrix $D$ agrees with the elements of $A$,
\begin{equation}
\forall_{i_1, \cdots, i_m }~~ D_{\; i_1 \; i_2  \cdots \; i_m}^{i_1 \; i_2 \; \cdots i_m} = A_{i_1 i_2 \cdots i_m}~.
\end{equation}
\end{defn}

Note that we do not demand here the $m$-stochasticity of the channel $\Phi_D$, relaxing this property from now on. This is because the $m$-stochasticity of quantum channels imposes rather strong constraints.
For instance, in the case $N = 2$, the only tristochastic coherification of the permutation tensor,

\begin{equation}
\label{simplest_pertm_A}
T_2  = (\id, P_2) = \left(\begin{matrix}
1 & 0 \\
0 & 1 \\
\end{matrix}\right.
\left|\begin{matrix}
0 & 1 \\
1 & 0 \\
\end{matrix}\right)~,
\end{equation}
is trivial in the sense that the dynamical matrix of the quantum channel $\Phi_D$ remains diagonal:

\begin{equation}
\label{simplest_pertm_D_diag}
D_{T,diag} = \begin{pmatrix}
  \begin{matrix}
  1 & 0 & 0 & 0 \\
  0 & 0 & 0 & 0 \\
  0 & 0 & 0 & 0 \\
  0 & 0 & 0 & 1 \\
\end{matrix}
  & \rvline & \bigzero \\
\hline
  \bigzero & \rvline &
  \begin{matrix}
  0 & 0 & 0 & 0 \\
  0 & 1 & 0 & 0 \\
  0 & 0 & 1 & 0 \\
  0 & 0 & 0 & 0 \\
  \end{matrix}
\end{pmatrix}~.
\end{equation} 
Therefore such a "classical" channel cannot give any "quantum advantage" over the standard classical convolution. Thus from now on we focus on coherifications of a classical convolution, represented by non diagonal dynamical matrix $D$, without demanding its tristochasticity on the quantum level.

An explicit formula for a convolution obtained by a coherification of the matrix \eqref{simplest_pertm_D_diag} is provided in Appendix \ref{app:explicit_form}.

Even though coherification $D$ of an $m$-stochastic tensor $A$ is in general not an $m$-stochastic channel, one may extract some information about generalized eigenvectors $\rho$ of $\Phi_D$ using only $m$-stochasticity of $A$, at least in the case of permutation tensors.

\begin{lem}\label{perumtation_coherification}
For any coherification $\Phi_D$ of an $m$-stochastic permutation tensor $T$, for each input states $\rho_1, \rho_2, \cdots$ the diagonal elements of $\rho = D[\rho_1, \rho_2, \cdots]$ depends only on diagonal elements of $\rho_1, \rho_2, \cdots$ by 
\begin{equation*}
\Phi_D[\rho_1, \cdots,\rho_{n-1}]_{\;i}^{i\;} = T_{i k_2 \cdots k_n} \;\; \rho_{1,\;k_2}^{\;\;k_2} \cdots \rho_{n-1,\;k_n}^{\;\;k_n} 
\end{equation*}
\end{lem}

\begin{proof}
To prove this statement it is sufficient to show that off diagonal terms of $\rho_1, \rho_2, \cdots$ do not affect the diagonal terms of $D[\rho_1, \rho_2,\cdots]$, which is equivalent to:
\begin{equation*}
D_{\; i \; k_2 \cdots \; k_n}^{i\;l_2 \cdots\;l_n} = 0 \text{ if } k_2 \neq l_2 \text{ or } k_3 \neq l_3 \text{ or } \cdots  \text{ or } k_n \neq l_n, 
\end{equation*}
where $D$ is the dynamical matrix of the channel $\Phi_D$ written in the tensor form.
Using the positivity of $D$, calculating the determinant of a minor,
\begin{equation*}
\begin{pmatrix}
 D_{\; i \; k_2 \cdots \; k_n}^{i\;k_2 \cdots\;k_n} &  D_{\; i \; k_2 \cdots \; k_n}^{i\;l_2 \cdots\;l_n}\\[10 pt]
 D_{\; i \; l_2 \cdots \; l_n}^{i\;k_2 \cdots\;k_n} &  D_{\; i \; l_2 \cdots \; l_n}^{i\;l_2 \cdots\;l_n}
\end{pmatrix}~,
\end{equation*}
one obtains an inequality
\begin{equation}
\label{coher_cond1}
    |D_{\; i \; k_2 \cdots \; k_n}^{i\;l_2 \cdots\;l_n}|^2 \leq  D_{\; i \; k_2 \cdots \; k_n}^{i\;k_2 \cdots\;k_n} D_{\; i \; l_2 \cdots \; l_n}^{i\;l_2 \cdots\;l_n} = T_{i k_2 \cdots k_n} T_{i l_2 \cdots l_n}~.
\end{equation}
Moreover, we must also invoke the trace preserving property of $D$, which implies, that for any $k_2, \cdots,k_n,l_2,\cdots,l_n$, if at least one pair of $k_r$, $l_r$ disagree,
\begin{equation}
\label{coher_cond2}
     0 = \sum_{i'} D_{\; i' \; k_2 \cdots \; k_n}^{i'\;l_2 \cdots\;l_n}~.
\end{equation}
In order to have a nonzero value of certain $D_{\; i \; k_2 \cdots \; k_n}^{i\;l_2 \cdots\;l_n}$, there must exist at least one more index $j$ such that $D_{\; j \; k_2 \cdots \; k_n}^{j\;l_2 \cdots\;l_n} \neq 0$. Which by \eqref{coher_cond1} imply that all $4$ elements of permutation tensor $ T_{i k_2 \cdots k_n}, T_{i l_2 \cdots l_n},  T_{j k_2 \cdots k_n}, T_{j l_2 \cdots l_n}$ must be nonzero. This condition cannot be satisfied because a permutation tensor cannot have  both entries $T_{i k_2 \cdots k_n}$ and $T_{j k_2 \cdots k_n}$ nonzero.
\end{proof}

\begin{thm}
Let $T$ be an $m$-stochastic permutation tensor of dimension $N$. Then for any coherification $\Phi_D$, the fixed points (generalized eigenvectors to eigenvalue $1$) of $\Phi_D$ inside $\Omega_N$ have a form
\begin{equation}
\rho_{\;i}^i = \frac{1}{N-k} \text{ for } i \notin I ~\text{and }~ \rho_{\;j}^i = 0  \text{ for } i \in I \text{ or } j \in I, 
\end{equation}
where the $I$ is a set of indices with respect to which the permutation tensor $A$ is reducible and $k = \# I$.
\end{thm}

\begin{proof}
By Lemma \ref{perumtation_coherification} the diagonal elements of the fixed point of $\Phi_D$, $\rho_{\;i}^i$, must be equal to the coefficients of some generalized eigenvectors of $T$: $p_i$. Moreover, Theorem \ref{m_stoch_tensor_finale} implies that $p_i = 0$ if $i \in I$ and $p_i = \frac{1}{N-k}$ if $i \notin I$, hence the values of diagonal terms of $\rho$ follows. The off-diagonal terms $\rho_{\;i}^j$ for $i \in I$ or $j \in I$ are equal $0$ due to positivity of $\rho$, since at least one of the diagonal terms $\rho_{\; i}^i$ or $\rho_{\;j}^j$ in the minor:
\begin{equation*}
\begin{pmatrix}
\rho_{\;i}^i & \rho_{\;i}^j\\
\rho_{\;j}^i & \rho_{\;j}^j
\end{pmatrix}
\end{equation*}
is equal to zero.
\end{proof}

\section{Optimal coherification of tristochastic permutation tensors }\label{sec:coch}

The aim of this section is to identify convolution between quantum states determined by optimal coherification of tristochastic tensors. The main quantity which we choose to maximize is the $2$-norm coherification $C_2$ \cite{Korzekwa_coherifying}, which quantifies the average contribution of the non-diagonal entries of the dynamical matrix $D$,

\begin{equation}
\label{C2_def_eq}
\begin{aligned}
& C_2(\Phi_T) = \sum_{kmln}|(\rho_\Phi)^{k\;n\;}_{\;l\;m}|^2 - \sum_{kn}|(\rho_\Phi)^{k\;n\;}_{\;k\;n}|^2 = \frac{1}{N^4}\sum_{kmln}|(D_T)^{k\;n\;}_{\;l\;m}|^2 - \frac{1}{N^4}\sum_{kmln}|(D_{T,diag})^{k\;n\;}_{\;l\;m}|^2 \\
& = \frac{1}{N^4} \sum_{\mu} \lambda_{D,\mu}^2 - \frac{1}{N^4}\sum_{\nu} \lambda_{D_{T,diag},\nu}^2~.
\end{aligned}
\end{equation}
Here $\Phi_T$ is a given coherification of a tristochastic tensor $T$, $\lambda_{D_{\text{diag}},\mu}$ are eigenvalues of dynamical matrix $D_T$ and $\lambda_{D_{T,\text{diag}},\mu}$ are eigenvalues of the dynamical matrix of diagonal coherification - see example \eqref{simplest_pertm_D_diag}. Expression \eqref{C2_def_eq} without extracting diagonal terms is also called the purity of a channel $\gamma(\Phi_T)$. The first pair of indices $k,l$ run over the set $\{1,\cdots,N\}$, whereas the second one $n,m$ over the set $\{1,\cdots,N^2\}$.

The difference between the entropies of the rescaled Choi dynamical matrix and of its diagonal leads to another, entopic measure of coherence \cite{Korzekwa_coherifying} $C_e$:
\begin{equation}
C_e(\Phi_D) = S(D_T/N^2) - S(D_{T,diag}/N^2),  
\end{equation}
where $S$ is von Neumann entropy and the prefactor $N^2$ comes from the normalization of the Jamiołkowski state, $\rho_{\Phi} = D/N^2$. 

\subsection{Coherification of the tristochastic permutation tensor of dimension $2$}

We start by discussing the coherification of the permutation tensor \eqref{simplest_pertm_A}, rewriting the channel $\Phi_T$ in the Kraus representation,

\begin{equation}
    \Phi_T(\rho) = \sum_{i=1}^{k} K_i^\dagger \rho K_i~.
\end{equation}
Since $\Phi_T$ is a map $\Omega_2 \otimes \Omega_2 \to \Omega_2$, each Kraus operator $K_i$ is represented by a rectangular matrix of size $2 \times 4$.
The dynamical matrix $D_T$ can be expressed by entries of Kraus operators:
\begin{equation*}
    (D_T)_{\; b \; d}^{a \;c \;} = \sum_i^k (K_i)_{\; c}^a  \overline{(K_i)_{\;d}^b}~.
\end{equation*}
Where $k$ is the number of Krauss operators.
For further details consult Appendix \ref{Kraus_app}.
Comparing the diagonal elements we immediately get $ \left( D_T \right)^{a \; c}_{\; a \; c} = \sum_i |(K_i)^a_{\; c}|^2 $, hence  the Kraus operators must be of the form,
\begin{equation}
    K_i = \begin{pmatrix}
        a_i & 0 & 0 & d_i \\
        0 & f_i & g_i & 0
        \end{pmatrix}~,~~i = 1, \cdots, k
\end{equation}
with $||a||^2 = ||b ||^2 = ||c||^2 = ||d||^2 = 1$. 
Further simplification comes form the trace preserving property, $\Tr_1 D_T = \id$, equivalent to  $\sum_{i=1}^{k} K_i^\dagger K_i = \id$. This implies
\begin{equation}
    \left(
\begin{array}{cccc}
 || a|| ^2 & 0 & 0 & \la a| d \ra \\
 0 & || f||^2 &\la f| g \ra& 0 \\
 0 &\la g| f \ra & || g|| ^2 & 0 \\
 \la d| a \ra & 0 & 0 & || d|| ^2 \\
\end{array}
\right) = 
\left(
    \begin{array}{cccc}
     1 & 0 & 0 & 0 \\
     0 & 1 & 0 & 0 \\
     0 & 0 & 1 & 0 \\
     0 & 0 & 0 & 1 \\
    \end{array}
    \right)~.
    \label{matseq}
\end{equation}
Thus four scalar products are equal to zero and the remaining free terms are $z_1 = \la f|a \ra$, $z_2 = \la g|a\ra$, $z_3 = \la f|d\ra$, $z_4 = \la g|d\ra$ and their conjugates.
The resulting dynamical matrix $D_T$  for $T$ from \eqref{simplest_pertm_A} reads,
\begin{equation}
\label{bigchoi}
D_T = \left(
\begin{array}{cccccccc}
 1 & 0 & 0 & 0 & 0 & z_1 & z_2 & 0 \\
 0 & 0 & 0 & 0 & 0 & 0 & 0 & 0 \\
 0 & 0 & 0 & 0 & 0 & 0 & 0 & 0 \\
 0 & 0 & 0 & 1  & 0 & z_3 & z_4 & 0 \\
 0 & 0 & 0 & 0 & 0 & 0 & 0 & 0 \\
 \bar{z_1} & 0 & 0 & \bar{z_3} & 0 & 1 & 0 & 0 \\
 \bar{z_2} & 0 & 0 & \bar{z_4} & 0 & 0 &  1 & 0 \\
 0 & 0 & 0 & 0 & 0 & 0 & 0 & 0 \\
\end{array}
\right)~.
\end{equation}
In this case, the expression for purity of the corresponding channel $\Phi_T$ reduces to
\begin{equation}
\label{gamma_phi_eq}
    \gamma(\Phi_T) = \frac{1}{16} \left[ 4 + 2 \left( |z_1|^2 + |z_2|^2 + |z_3|^2 + |z_4|^2 \right) \right]~.
\end{equation}
Notice that $|f\ra$ and $|g \ra$ are two orthonormal vectors, therefore one can span an orthonormal basis including those two vectors $|f\ra = |e_1\ra$ and $|g\ra = |e_2\ra$. The number of basis vectors depends on the number $k$ of Kraus operators $K_i$. Hence the term with the inner products in \eqref{gamma_phi_eq} can be rewritten as
\begin{equation*}
|z_1|^2 + |z_2|^2 + |z_3|^2 + |z_4|^2  =  |a_1|^2 + |a_2|^2 + |d_1|^2 + |d_2|^2 = 1 - \sum_{i=3}^k |a_i|^2 + 1 - \sum_{i=3}^k |d_i|^2 \leq 2~.
\end{equation*}
Here $a_i$, $d_i$ denote the components of vectors $|a\ra$, $|b\ra$ in the  basis mentioned above.
Thus we found an upper bound on the channel purity $\gamma(\Phi_T)$, which can be achieved using $k = 2$ Kraus operators.
The maximal value of 2-norm coherence is,
\begin{equation}
\label{bound}
    \mathcal{C}_2(\Phi_T) = \gamma(\Phi_T) - \frac{1}{16} \operatorname{Tr}(T T^\dagger) = \gamma(\Phi_T) - \frac{1}{4} \leq \frac{1}{16} \left( 4 + 2 \cdot 2\right) - \frac{1}{4} = \frac{1}{4}~.
\end{equation}

In Appendix \ref{app:explicit_form} we present the final expression for the optimal coherification $\Phi_T$ of a tensor $T$.
Moreover, in  Appendix \ref{app:qutib_gen_coch} we propose our candidate for the optimal coherification of any $N = 2$ dimensional tristochastic tensor with respect to the $2$-norm coherification \eqref{C2_def_eq}.

\subsection{Coherification of permutation tensors}

Knowing how to construct an optimal coherification of the permutation tensor of dimension $2$, we can generalize this method for an arbitrary permutation tensor of dimension $N$.

Let us start with the condition $\sum_i |(K_i)_{j}^{\;kl}|^2 = T_{jkl}$, which implies that each Kraus operator has at most $N^2$ nonzero entries. Let us denote the value of $n^{\text{th}}$ nonzero entry in $j^{\text{th}}$ row of $i^{\text{th}}$ Kraus operator as $B_{i n}^j$, with scalar product between them understood as $\la B_n^j | B_{n'}^{j'} \ra = \sum_i \bar{B}_{i,n}^j B_{i,n'}^{j'}$. For example, for the permutation tensor \eqref{example_permutation_tensor} of order three, the Kraus operators are parametrized as
\begin{equation}
\label{eq41}
K_i = \left(\begin{matrix}
B_{i1}^{1} & 0 & 0 & 0 & B_{i2}^{1} & 0 & 0 & 0 & B_{i3}^{1}\\
0 & B_{i1}^{2} & 0 & 0 & 0 & B_{i2}^{2} & B_{i3}^{2} & 0 & 0\\
0 & 0 & B_{i1}^{3} & B_{i2}^{3} & 0 & 0 & 0 & B_{\;i3}^{3} & 0\\
\end{matrix}\right)~.
\end{equation}

Since $T$ is a permutation tensor there is only a single nonzero element in each column in $K_i$. Hence checking the trace preserving condition, $K_i^\dagger K_i = \id$, we obtain $||\;|B_n^j\ra||^2 = 1$ on the diagonal, and scalar products $\la B_{n_1}^j| B_{n_2}^j \ra = 0$, outside the diagonal.
Therefore, for each $j$ the set $\{|B_n^j \ra\}_{n = 1}^N$ is an orthonormal set with $N$ elements and the minimal number of Kraus operators is $N$.

Calculation of $2$-norm coherence  $C_2(\Phi)$ boils down to a sum of squared norms of projections of vectors onto elements of a certain basis: 

\begin{equation}
\begin{aligned}
& C_2(\Phi_D) = \frac{1}{N^4}  \sum_{klab}|(D_T)^{k\;a\;}_{\;l\;b}|^2  - \frac{1}{N^4} \sum_{ka}|(D_T)^{k\;a\;}_{\;k\;a}|^2 =  \\
& ~~= \frac{1}{N^4} \sum_{k,n,l,m} |\la B_n^k|B_m^l  \ra|^2 - \frac{1}{N^4} \sum_{k,n} |\la B_n^k|B_n^k  \ra|^2 = \frac{1}{N^4} \sum_{k,n,l,m} |\la B_n^k|B_m^l  \ra|^2 - \frac{1}{N^2}.
\end{aligned}
\end{equation}
As in the case of qubits, the result is maximal, if all vectors $|a_{(k,n)}\ra$ belong to the same $N$ dimensional space, so none of their components in any basis would be lost. In such a case the resulting coherence of the channel $\Phi_T$ reads,

\begin{equation}
\begin{aligned}
& C_2(\Phi_T) = \frac{1}{N^4} \sum_{k, l,m} \sum_{n}  |\la B_n^k|B_m^l  \ra|^2  - \frac{1}{N^2}= \frac{1}{N^4} \sum_{k,l,m} 1 - \frac{1}{N^2} =  \frac{N-1}{N^2} ~.
\end{aligned}
\end{equation}

To calculate the maximal entropic coherence, first note that the eigenvalues of $\rho_{\Phi}$ are equal to
\begin{equation*}
\frac{1}{N^2}|| K_i||^2 =\frac{1}{N^2} \sum_{j} \sum_{n}|\la i |B_n^j \ra|^2 =\frac{1}{N^2} \sum_{j}  || \;|i\ra ||^2  =\frac{1}{N^2} \sum_{j} 1 = \frac{1}{N}~,
\end{equation*}
Where in the second step we used the fact that $\{|B_n^j \ra\}_{n = 1}^N$ are orthogonal bases. Therefore the optimal entropic coherence is equal to

\begin{equation*}
C_2(\Phi) = S(\text{diag}(\rho_{\Phi})) -  S(\rho_{\Phi}) = - \ln(1/N^2)  + \ln(1/N)  = \ln(N) ~.
\end{equation*}

\section{Convolution of quantum states}\label{sec:Qconvolution}

Having understood the structure of coherifications for permutation tensors, we can move on to their implementation.
The Theorem presented below serves as a helpful tool in achieving this goal.

\begin{thm}
\label{lem_channel_generation}
Let $B^{k}$ denote unitary matrices, whose rows correspond to basis vectors in $\fC^N$.
Then for each coherification $\Phi_T$ of a permutation tensor $T$ there exist a unitary matrix of the form $U = B P$, where $P$ is a permutation matrix of dimension $N^2$ and $B$ is a block diagonal matrix with diagonal given by a simple sum of $N$ blocks, $B = B^1\oplus \cdots \oplus B^N$, such that:
\begin{equation}
\label{eq51}
\rho_1 \star_T \rho_2  := \Phi_T[\rho_1, \rho_2]= \Tr_2[U(\rho_1 \otimes \rho_2)U^\dagger] ~,
\end{equation}
\end{thm}

\begin{proof}
The coherification of the permutation tensor gives exactly $N$ Kraus operators of the form \eqref{eq41}, hence we can define unitary operator $U = \sum_i K_i \otimes |i\ra$, so that:

\begin{equation}
\label{eq52}
\Tr_2[U(\rho_1 \otimes \rho_2)U^\dagger] = \Tr_2[(\sum_i K_i \otimes |i\ra)(\rho_1 \otimes \rho_2)(\sum_j K_j^\dagger \otimes \la j|)] = \sum_i K_i (\rho_1 \otimes \rho_2) K_i^\dagger = \Phi_T[\rho_1,\rho_2] ~.
\end{equation}

In the case of the permutation tensor $T_3$ of order three defined in eq. \eqref{example_permutation_tensor} the matrix $U$ of order nine has the form:

\begin{equation}
\label{eq53}
U = \left(\begin{matrix}
B_{11}^{1} & 0 & 0 & 0 & B_{12}^{1} & 0 & 0 & 0 & B_{13}^{1}\\
B_{21}^{1} & 0 & 0 & 0 & B_{22}^{1} & 0 & 0 & 0 & B_{23}^{1}\\
B_{31}^{1} & 0 & 0 & 0 & B_{32}^{1} & 0 & 0 & 0 & B_{33}^{1}\\
0 & B_{11}^{2} & 0 & 0 & 0 & B_{12}^{2} & B_{13}^{2} & 0 & 0\\
0 & B_{21}^{2} & 0 & 0 & 0 & B_{22}^{2} & B_{23}^{2} & 0 & 0\\
0 & B_{31}^{2} & 0 & 0 & 0 & B_{32}^{2} & B_{33}^{2} & 0 & 0\\
0 & 0 & B_{11}^{3} & B_{12}^{3} & 0 & 0 & 0 & B_{13}^{3} & 0\\
0 & 0 & B_{21}^{3} & B_{22}^{3} & 0 & 0 & 0 & B_{23}^{3} & 0\\
0 & 0 & B_{31}^{3} & B_{32}^{3} & 0 & 0 & 0 & B_{33}^{3} & 0\\
\end{matrix}\right)~,
\end{equation}
where $B_{in}^k$ are simultaneously the coefficients from eq. \eqref{eq41} and the matrix elements of $B^k$.
In general, for a permutation tensor $T$, the unitary $U$ is of the form:
\begin{equation}
U_{ki,lj} = A_{klj} B_{il}^{k}~.
\end{equation}

The last step is the observation that there always exists a permutation of columns, such that the $U P^\top = B$ achieves the desired block structure.
    
\end{proof}

\begin{rmk}
In general there, are many equivalent sets of Kraus operators defining the same channel, hence to find only the relevant information describing any channel one should rather consider its dynamical matrix $D$.
The values of the dynamical matrix contain only scalar products $\la B_{n}^k, | B_{m}^l\ra$. Thus one can rotate all vectors $|B_{n}^k\ra$ without affecting a channel action. Therefore in Theorem \ref{lem_channel_generation} we may fix the first diagonal block $B^1 =\mathbb{I}_N$.
\end{rmk}

\subsection{Channel generation and quantum circuit for qubits}

As an illustrative example let us examine the coherification of the simplest tristochastic tensor \eqref{simplest_pertm_A}.

Note that the dynamical matrix of the channel \eqref{bigchoi_2} depends only on the scalar product between vectors $|a\ra$, $|d\ra$, and $|f\ra$, $|g\ra$. Therefore we might rotate the parameter vectors $|a\ra$, $|d\ra$ by any unitary matrix $u$ of order two, and $|f\ra$, $|g\ra$ by $u^\dagger$ without affecting the action of the channel. Using this freedom we set $|a\ra = |1\ra$ and $|b\ra = |2\ra$. Hence the Kraus operators for any channel in this set are defined by only three phases:\\ $\alpha \in [0,2 \pi]$, $\theta \in [0,\pi]$ and $\phi \in [0,2\pi]$:
\begin{equation}
\label{two_q_channel_krauss}
    \begin{gathered}
    K_1 = \begin{pmatrix}
        1 & 0 & 0 & 0 \\
        0 & e^{i \alpha} \cos{\frac{\theta}{2}} & e^{i \alpha} \sin{\frac{\theta}{2}} & 0
        \end{pmatrix}~,
    \quad
    K_2 = \begin{pmatrix}
        0 & 0 & 0 & 1 \\
        0 & e^{i \alpha}e^{i \phi} \sin{\frac{\theta}{2}} & -e^{i \alpha}e^{i \phi} \cos{\frac{\theta}{2}} & 0
        \end{pmatrix}~~
    \end{gathered}
\end{equation}

Repeating the steps from eq. \eqref{lem_channel_generation} we get the unitary matrix $U_4$,

\begin{equation}
\label{two_q_channel}
    \begin{aligned}
        U_4 = K_1 \otimes |1\ra + K_2 \otimes |2\ra =
        \left(
            \begin{array}{cccc}
            1 & 0 & 0 & 0 \\
            0 & 0 & 0 & 1 \\
            0 & e^{i \alpha} \cos{\frac{\theta}{2}} & e^{i \alpha} \sin{\frac{\theta}{2}} & 0 \\
            0 & e^{i \alpha}e^{i \phi} \sin{\frac{\theta}{2}} & -e^{i \alpha}e^{i \phi} \cos{\frac{\theta}{2}} & 0 \\
        \end{array}
        \right)~.
    \end{aligned}
\end{equation} 

To design the corresponding circuit we  decompose an operation $U_4$ into gates from the universal set \cite{Nielsen_Chuang}
\begin{equation}
    \Qcircuit @C=1.0em @R=.7em {
    & \qw & \multigate{1}{U_4(\alpha,\theta,\phi)} & \qw & \raisebox{-2.8 em}{=} &  & \qw &
    \gate{H} & \ctrl{1} & \gate{H} & \ctrl{1} & \multigate{1}{\Lambda(\alpha,\theta,\phi)} & \qw \\
    & \qw & \ghost{U_4 (\alpha, \theta, \phi)} & \qw & & & \qw & \gate{H} & \targ    & \gate{H} & \targ    & \ghost{\Lambda(\alpha,\theta,\phi)}        & \qw \\
    }~,
\end{equation}
where the gate $\Lambda(\alpha,\theta,\phi)$ can be further decomposed  by  $Z(\alpha) =  |0\ra\la0| + e^{i\alpha} |1\ra\la 1| $  and $X(\alpha) = |+\ra\la +| + e^{i\alpha} |-\ra\la -| $~,

\begin{equation}
    \Qcircuit @C=1.0em @R=.7em {
    & \qw & \multigate{1}{\Lambda(\alpha,\theta,\phi)} & \qw & \raisebox{-2.8 em}{=} & & \qw & \gate{Z\left(\alpha - \frac{\theta}{2}\right)} & \qw & \ctrl{1} & \qw & \ctrl{1} & \qw & \ctrl{1} & \qw \\
    & \qw & \ghost{\Lambda(\alpha,\theta,\phi)}        & \qw & & & \qw & \qw & \qw & \gate{Z\left(\frac {\pi}{2}\right)} & \qw & \gate{X\left(\theta\right)} & \qw & \gate{Z\left(\phi+\frac {\pi}{2}\right)} & \qw\\
    }~.
\end{equation}

The gate $U_4$ can be  decomposed in an alternative, more transparent way,

\begin{equation}
    \Qcircuit @C=1.0em @R=.7em {
     & \qw & \multigate{1}{U_4 (\alpha, \theta, \phi)} & \qw & \raisebox{-2.8 em}{=} & & \qw & \gate{e^{-i \phi /2}} & \multigate{1}{Q.CONV (\theta)} & \gate{e^{i \phi /2}} & \gate{e^{i \alpha}} & \qw \\
    &  \qw & \ghost{U_4 (\alpha, \theta, \phi)} & & & & \qw &  \gate{e^{-i \phi /2}} & \ghost{Q.CONV (\theta)} & \qw & \qw &\qw \\
    }~,
\end{equation}
where $Q.CONV (\theta) = U_4(0, \theta, 0)$ is a two-qubit gate, dependent only on the parameter $\theta$.
Therefore the parameter $\alpha$ is an additional global phase of the resulting state and often can be set to $0$. 
The role of the parameter $\phi$ is similar and it corresponds to applying a $- \phi /2$ phase gate on both states before the convolution and $+ \phi /2 $ phase gate after it.
Thus only the angle $\theta$ genuinely influences the way how the convolution acts.

\begin{figure}[h]
    \centering
        \includegraphics[height=4in]{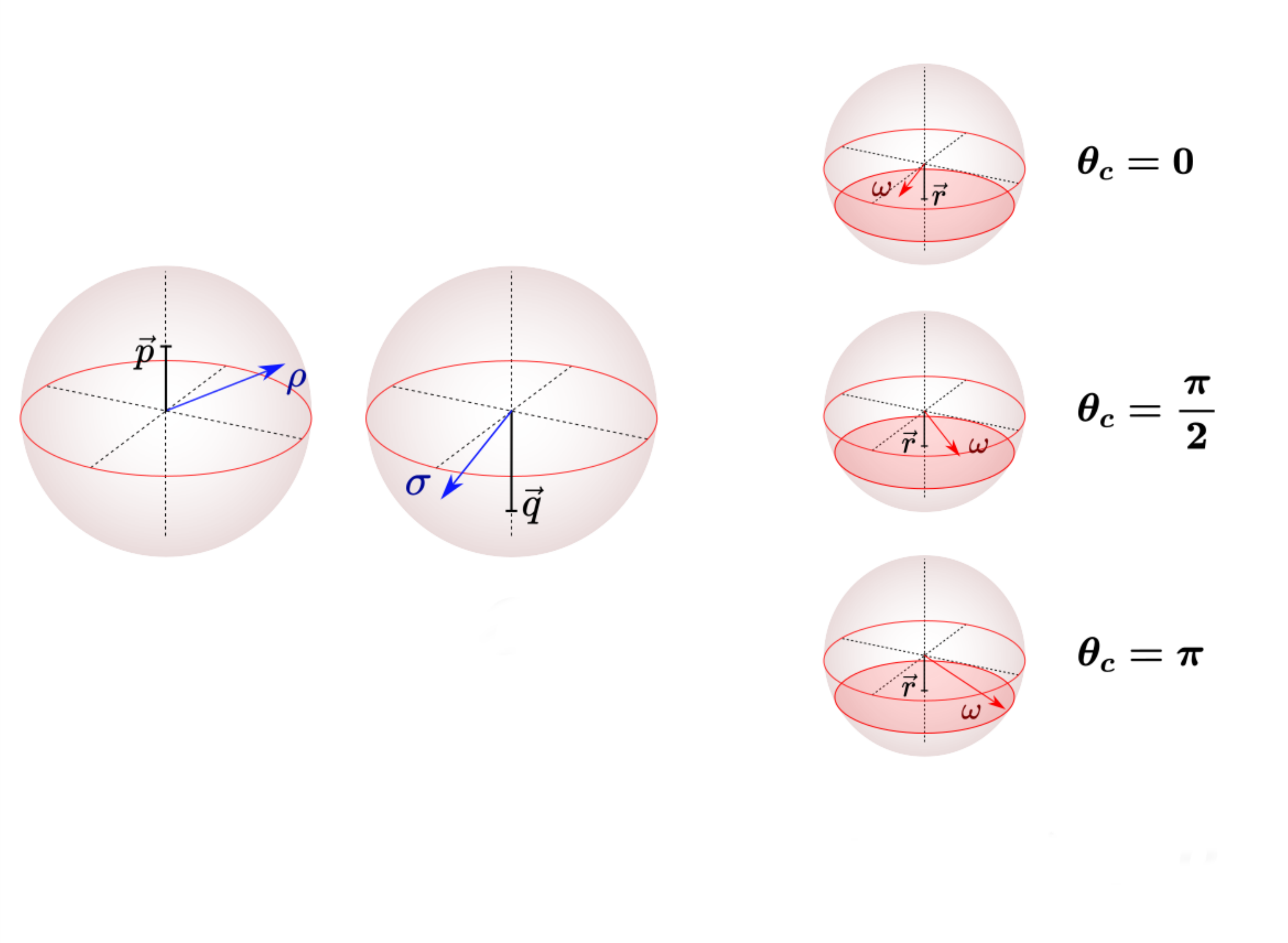}
     \centering
     \caption{\label{fig:action} The action of two qubit convolution \eqref{two_q_channel},  $\omega = \rho \star_T \sigma$,  with parameters $\alpha$, $\phi$ set to $0$, for pure states $\rho$, $\sigma$ visualized on the Bloch ball by blue vectors.
     Three versions of the convolution operator parametrized by the angle $\theta = 0, \frac{\pi}{2}$ and $\pi$ are shown in red, and the decoherent input $\vec{p}$, $\vec{q}$ and output $\vec{r}$ states are denoted by black vertical segments. The red plane marks all possible outputs of $U_{4}(\alpha,\theta,\phi)$, for given input states. }
\end{figure}

\subsection{Further study of qubit convolution}

The choice of the optimal convolution \eqref{two_q_channel} depends on the way, how these operators will be used. An example, proposed in Ref. \cite{QCNN1}  concerns coding against the correlated noise. 

Notice that if we encode one-bit values $0 \to 00$ (or $11$) and $1 \to 10$ (or $01$), and allow for action of the noise that may flip simultaneously both bits, the permutation tensor \eqref{simplest_pertm_A} recovers the original value.
A similar procedure can be performed on the quantum level. Firstly we enlarge the Hilbert space of the system
\begin{equation*}
\rho \to \rho \otimes |0\ra\la 0 |~.
\end{equation*}
Since matrix $U_4$ corresponds to decoding and error correction we use $U_4^\dagger$ as an encoding, so that with an absence of noise their joint action compensates. 

Consider a noise of the form $R(\vec{r})\otimes R(\vec{r}) $, where $R(r)$ is a rotation of a qubit along a vector $r$, with a phase proportional to $|r|$:
\begin{equation*}
R(\vec{r}) = \exp(i \frac{\pi}{2} \left(r_1 \sigma_1 + r_2 \sigma_2 + r_3 \sigma_3\right) ) = \cos\left(\frac{\pi|r|}{2}\right)\id + i \sin\left(\frac{\pi|r|}{2}\right) \left(\hat{r} \vec{\sigma}\right).
\end{equation*}
Operation of encoding, transformation by the noisy channel and decoding are given by the following evolution of a given state $\rho$, 
\begin{equation}
\label{eq61}
\rho \to \Phi(\rho) := \Tr_2\left[\left(U (R(r)\otimes R(r)) U^\dagger\right)(\rho \otimes |0\ra\la0|)\left(U (R(r)^\dagger\otimes R(r)^\dagger) U^\dagger\right) \right].
\end{equation}

Considering the action of \eqref{eq61} we  calculate the fidelity between the input $\rho$ and output $\Phi(\rho)$, which for the pure state $\rho = |\psi\ra\la\psi|$ reads
\begin{equation*}
 F(|\psi\ra\la\psi|, \Phi(|\psi\ra\la\psi|)) = \la \psi|\Phi(|\psi\ra\la\psi|)|\psi\ra~.
\end{equation*}
Such a value describes how much noise diminishes our capabilities to recognise the outcome state.
For the basis states $|0\ra$ and $|1\ra$ on gets the following values
\begin{equation*}
\begin{aligned}
& F(|0\ra\la 0|, \Phi(|0\ra\la0|))  = w^2 \sin ^4\left(\frac{\pi  r}{2}\right)+\frac{1}{4}\left(w \cos (\pi  r)+\hat{r}_3^2+1\right)^2~,\\
& F(|1\ra\la 1|, \Phi(|1\ra\la1|)) = 1- \frac{1}{2} w (\sin (\theta )+1) \left(\hat{r}_3^2 (\cos (\pi  r)-1)^2+\sin ^2(\pi  r)\right),
\end{aligned}
\end{equation*}
with $w = 1 -  \hat{r}_3^2$. The results depend only on a single parameter describing the matrix $U_4$, the angle $\theta$, appearing only in the second expression. The optimal value $\theta_{opt} = -\frac{\pi}{2}$ corresponds to:
\begin{equation*}
\begin{aligned}
& F(|1\ra\la 1|, \Phi(|1\ra\la1|))|_{\theta = -\frac{\pi}{2}}   = 1 ~.\\
\end{aligned}
\end{equation*}

Another, more abstract, way to determine the influence of parameters $\alpha$, $\theta$, $\phi$ in the qubit convolution is to examine entangling power \cite{ep_Zanardi} and gate typicality \cite{JMZL_entanglement_measures} of any two-qubit unitary gate $U_4$.

Entangling power $e_p \in [0,1]$ of a channel $U$ is a quantity describing, how much the outcome $U (|\psi\ra \otimes |\phi\ra)$ is entangled on average for random input pure states $|\psi\ra$ and $|\phi\ra$. After the partial trace in \eqref{eq51} it gives us insight how much the result of the convolution becomes mixed. The entangling power of $U_4$ achieves the maximal value of qubit channels equal$e_(U_4) = 2/3$, independent of the parameters $\theta, \phi, \alpha$. 
Another measure, the gate typicality $g_t \in [0,1]$, specifies how much input states have been "interchanged" during the action of the unitary channel and for $U_4$ obtain values:

\begin{equation*}
    g_t(U_2) = \frac{1}{6}\left(3 - \cos \theta  \right).
\end{equation*}
Hence for $\theta = 0$ more information from the first state is lost during the partial trace, and for $\theta = \pi$ more information from the second state is lost during the partial trace. The case $\theta = \pm \pi/2$ corresponds to the "symmetric treatment" of both input states which can be desired.
Detailed discussion of entangling power and gate typicality is provided in Appendix \ref{app:ep_gt}.
Our numerical simulations also suggest that setting $\theta = \pm \frac{\pi}{2}$ guarantees the slowest rate of entropy increase in multiple convolution schemes.

Note that for $\theta = \pm \frac{\pi}{2}$ the basis $\{|a\ra, |d\ra \}$ and $\{|f\ra, |g\ra \}$ in \eqref{bigchoi} are mutually unbiased \cite{Wootters}. 
Thus, we presume that the most preferable coherification of an arbitrary permutation tensor usually has all the bases $\{|a_{(k,l)}\ra\}_{l = 1}^N$ mutually unbiased. Such a scheme is doable for any dimension $N$, which is a prime or a power of a prime \cite{Wootters}.

\section{Concluding remarks}

In this work, we presented a family of products $r = p \star_A q$, parametrized by a tristochastic tensor $A$,  defined on the set of classical probability vectors. Furthermore, we analysed the discrete dynamics induced by $m$-stochastic tensors. Investigations performed from the perspective of generalized Markov processes led to the characterization of the generalized eigenvectors. An alternative perspective of binary operations allowed us to study the connectivity, commutativity, and the existence of a neutral element and inverse elements for these operations.

The above notions and results were translated into the quantum setup. We analyzed tristochastic and multi-stochastic quantum operations and studied their properties.
In the next step, we enlarged a class of discussed channels defining \textit{coherification} \cite{Korzekwa_coherifying} of $m$-stochastic tensors to take full advantage of the quantum properties. We provided an explicit way to construct a coherification matrix $D \geq 0$ of tristochastic permutation tensors with maximal norm two coherence, which yields a constructive recipe to convolute arbitrary two quantum density matrices of the same dimension $\omega = \rho \times_D \sigma$ such that the coherence in preserved as much as possible.
Finally, we analyzed this class of convolutions for qubits, discussing their properties and possible applications.

Our results raise several questions worthy further study. First and foremost, the action of quantum $m$-stochastic tensors and coherification of $m$-stochastic channels should be examined for the case of entangled states as the input of the convolution operator, which is the more natural assumption in the case of convolutions in quantum convolution neural networks \cite{Hartmann}. Another useful topic is the collective behaviour of "interconnected" convolutions operating across multiple subsystems. This leads us also to the issue of parametrization and implementation of convolution, which was discussed here only in the simplest case of the convolution of two single qubit states.

\textit{Acknowledgements:} 

It is a pleasure to thank Kamil Korzewa, Zbigniew Puchała, Martin Seltmann, Fereshte Shahbeigi and Norbert Steczyński for fruitful discussions.
Financial support by NCN under the Quantera project no. 2021/03/Y/ST2/00193 and by Foundation for Polish Science under the Team-Net
project no. POIR.04.04.00-00-17C1/18-00 is gratefully acknowledged.

\newpage
\newpage
\appendix

\section{On binary quantum channels}\label{app:on_chanels}

In this Appendix, we recall the basic properties of binary quantum channels
\begin{equation*}
    \Phi_D: \Omega_N \otimes \Omega_N \to \Omega_N~,
\end{equation*}
that map two density matrices of order $N$ to another density matrix of the same size.
We temporarily denote these three Hilbert spaces as $\mathcal{A}$, $\mathcal{B}$ and $\mathcal{C}$ and the corresponding sets of density matrices as $\Omega_A$, $\Omega_B$, $\Omega_B$.
It suffices to give a definition of such a channel for separable states and extend it by linearity
\begin{equation}
    \Phi(\alpha \cdot \rho_0\otimes\sigma_0 +  \beta \cdot \rho_1\otimes\sigma_1) = \alpha \cdot \Phi(\rho_0\otimes\sigma_0) + \beta \cdot \Phi(\rho_1\otimes\sigma_1)~.
\end{equation}
Such maps can be represented in various ways. The following ones occur to be the most convenient for our work.

\subsection{Dynamical matrix representation via Choi-Jamiołkowski isomorphism}

In general, any completely positive trace preserving channel $\Phi$ between Hilbert spaces $\mathcal{X}$ and $\mathcal{Y}$ can be expressed using Choi-Jamiołkowski isomorphism \cite{Choi, Jamiolkowski} as
\begin{equation}
    \Phi (\rho) = \operatorname{Tr}_{X}(D (\mathbb{I}_Y \otimes \rho^\top ))~,
\end{equation}
with $D \geq 0$ and $\Tr_A[D] = \id$.
In this work, we study the generalized case of two inputs, $\mathcal{X} \rightarrow \mathcal{A} \otimes \mathcal{B}$; $\mathcal{Y} \rightarrow \mathcal{C}$. Then such an expression takes the form,

\begin{equation}
\label{binary_channel}
\rho \star_D \sigma := \Phi_D (\rho\otimes\sigma) = \operatorname{Tr}_{AB}(D (\mathbb{I}_C \otimes \rho^\top \otimes \sigma^\top )) ~.
\end{equation}

Any valid dynamical matrix, $D = D_{ABC}$, acting on the space $A\otimes B \otimes C$, induces a binary operation on quantum states $\rho \star_D \sigma  = \Phi_D (\rho\otimes\sigma)$. 
The necessary and sufficient conditions for trace preserving and complete positivity of such a channel in the Choi representation are
\begin{enumerate}
    \item $\Tr_{C}[D_{ABC}] = \mathbb{I}_{AB}$~,
    \item $D_{ABC}$ is a positive operator, $D_{ABC} \geq 0$ .
\end{enumerate}
The requirement of tristochasticity is just a repetition of the first condition for partial traces over the subsystems $A$ and $B$.

\subsection{Kraus representation}\label{Kraus_app}

Kraus operators for binary operation are linear maps between spaces $\mathcal{A} \otimes \mathcal{B}$ and $\mathcal{C}$, which means that they can be represented as $N \times N^2$ rectangular matrices. The action of a channel $\Phi$ in the Kraus representation takes the form.
\begin{equation*}
\Phi(\rho) = \sum_i K_i \rho K_i^\dagger~.
\end{equation*}
The condition of positivity is satisfied automatically, and the condition for trace preservation requires that $\sum_i K_i^\dagger K_i = \mathbb{I}$. A minimal number of Kraus operators necessary to represent the channel $r_\Phi$ is called the \textit{rank} of the channel.

To connect the Kraus representation with the dynamical matrix let us apply the Choi-Jamiołkowski. Let the indices ${\mu, \nu , c,d}$ run from 1 to $N^2$ and correspond to $\mathcal{A}\otimes \mathcal{B}$ space, while
${a,b}$ run from 1 to $N$ and correspond to $\mathcal{C}$ space. Moreover, let $E_{(\mu, \nu)}$ be a matrix with $1$ in the matrix entry $\mu, \nu$ and $0$ in all the others, then
\begin{equation}
\begin{aligned}
&\left( D \right)^{a \; c}_{\; b \; d} = \sum_{\mu, \nu } (\Phi_D(E_{(\mu, \nu)}))^{a}_{\; b} \cdot (E_{(\mu, \nu)})^c_{\; d} = \sum_{\mu, \nu } ( \sum_i K_i E_{(\mu, \nu)} K_i^\dagger  )^{a}_{\; b} \cdot (E_{\mu, \nu})^c_{\; d} = \\
&~~~~~~~~~=  \sum_{\mu, \nu, i} (K_i^\dagger)^a_{\; \mu} (K_i)^\nu_{\; b} \cdot \delta_{\mu c} \delta_{\nu d} = \sum_i (K_i)^a_{\; c} \overline{(K_i)^b_{\; d}} ~.\\
\end{aligned}
\end{equation}

\subsection{Unitary evolution in an enlarged space}

Any channel $\Phi: \ \mathcal{X} \rightarrow \mathcal{Y}$,
can be associated with an isometric transformation $V \in L(\mathcal{X},\mathcal{Y}\otimes \mathcal{Z})$, for certain auxiliary space $\mathcal{Z}$ by so called \textit{Stinespring representation},
\begin{equation}
\label{eq_V_uni}
    \Phi(\rho) = \operatorname{Tr}_{\mathcal{Z}}(V \rho V^\dagger)~.
\end{equation}
The minimum dimension of the auxiliary space $\mathcal{Z}$ is equal to the rank $r_\Phi$ of a channel $\Phi$.
For the channel with equal Hilbert spaces, $\mathcal{X}=\mathcal{Y}$, it is straightforward to rewrite this expression using
some unitary transformation, $U \in L(\mathcal{X}\otimes \mathcal{Z}, \mathcal{Y} \otimes \mathcal{Z})$, defined by the relation $U(|i\ra \otimes |0\ra) := V |i\ra$. One obtains then
\begin{equation}
\label{U_extension}
    \Phi_D(\rho) = \operatorname{Tr}_{Z} \left(U (\rho_X \otimes |0\ra_{\mathcal{Z}} \la0|) U^\dagger\right)~.
\end{equation}

In the case of binary channels \eqref{binary_channel} the description of their action can take various forms. One can treat both arguments in the symmetric way and the auxiliary space plays the role of the output. After unitary evolution and the partial trace over the input spaces $\mathcal{A}$ and $\mathcal{B}$ one obtains the form similarly to \eqref{U_extension}.
\begin{equation}
    \Phi_U (\rho, \sigma) = \rho \star_U \sigma := \operatorname{Tr}_{AB}(U(\rho_A \otimes \sigma_B \otimes \ket{0}_C \bra{0})U^{\dagger})
    \label{defu}
\end{equation}
where $U \in U\left(\mathcal{A}\otimes \mathcal{B}\otimes \mathcal{C} \right)$. This approach is valid for any channel $\Phi_D$ with any rank $r_{\Phi}$.

Alternatively, in a case in which the binary channel has the rank $r_\Phi =N$, another, more compact unitary representation becomes natural.
The isometry $V$ from \eqref{eq_V_uni} is then a unitary operator,
\begin{equation*}
U\in \  L(\mathcal{A}\otimes B, \mathcal{C}\otimes \mathcal{Z}) \equiv L(\H_{N^2}, \H_{N^2}),
\end{equation*}
and the channel can be represented as:
\begin{equation}
\label{steinu}
    \Phi_U (\rho, \sigma) = \rho \star_U \sigma := \operatorname{Tr}_{B}(U(\rho_A \otimes \sigma_B)U^\dagger)~~~.
\end{equation}

Note that the asymmetry concerning both arguments related to the partial trace over subsystem $\mathcal{B}$ is only apparent, as a dual form, corresponding to partial trace over subsystem $\mathcal{A}$ can also be written.

\section{Proofs of technical results}\label{app:proofs}

\subsection{Proofs and calculations form Section 2}\label{app:clasical_conv}

In this subsection, we present the proofs of the theorems stated in section 2, starting with Theorem \eqref{multisto_fixed}, which for convenience is restated here.

\begin{thm}
Let $\vec{q}^{(0)}$ be any point in the interior of probability simplex $\Delta_N$ and $\vec{p}^{(0)}$ be a point at the boundary of the probability simplex $\partial\Delta_N$ such that $\vec{q}^{(0)} = \alpha \vec{e} + (1 -\alpha)\vec{p}^{(0)}$ for some $\alpha \in [0,1)$, Next let us denote  sequences $\{\vec{q}^{(n)}\}$, $\{\vec{p}^{(n)}\}$ by $\vec{q}^{(n+1)} = A[\vec{q}^{(n)},\cdots, \vec{q}^{(n)}]$, $\vec{p}^{(n+1)} = A[\vec{p}^{(n)},\cdots, \vec{p}^{(n)}]$ for any $m$-stochastic map $A$. If $m > 2$ then the sequence $\{\vec{q}^{(n)}\}$ converges to $\vec{e}$.
\end{thm}

\begin{proof}
Let us derive the general formula for $\vec{q}^{(n)}$ using $\vec{p}^{(n)}$. Starting from $\vec{q}^{(1)}$ one ges, 
\begin{equation}
\label{q_1}
\begin{aligned}
& \vec{q}^{(1)} = A[\vec{q}^{(0)}, \cdots, \vec{q}^{(0)} ] = A[\alpha \vec{e} + (1 - \alpha) \vec{p}^{(0)},\cdots,\alpha \vec{e} + (1 - \alpha) \vec{p}^{(0)} ] = \\
& ~~~~~ = \sum_{k = 0}^{m-1} \alpha^k (1 - \alpha)^{m-1-k}   \overbrace{ A[\underbrace{\vec{e}, \cdots, \vec{e}}_{k \text{ times}},\underbrace{\vec{p}^{(0)},\cdots,\vec{p}^{(0)}}_{m-1-k  \text{ times} } ]+ \cdots }^{\binom{m-1}{k} \text{ terms with $\vec{e}$ appearing $k$ times}} = \\
& ~~~~~ = \sum_{k = 1}^{m-1} \alpha^k (1 - \alpha)^{m-1-k} \binom{m-1}{k}  \vec{e} + (1 - \alpha)^{m-1} \vec{p}^{(1)} = \\
& ~~~~~ = \left(1 - (1 - \alpha)^{m-1} \right)\vec{e} + (1 - \alpha)^{m-1} \vec{p}^{(1)}~,
\end{aligned}
\end{equation}
where the $4$-th equality follows from Lemma \ref{lem_mstoch1}.
Repeating the above calculations one finds that,
\begin{equation}
\label{q_n}
\vec{q}^{(n)} = \left(1 - (1 - \alpha)^{(m-1)^n} \right)\vec{e} + (1 - \alpha)^{(m-1)^n} \vec{p}^{(n)}~.
\end{equation}

Hence the norm of $\vec{q}^{(n)}$ is equal to:
\begin{equation}
\begin{aligned}
&||\vec{q}^{(n)} ||_2^2 = \\
& = \left(1 - (1 - \alpha)^{(m-1)^n} \right)^2 ||\vec{e}||_2^2 + 2 \left(1 - (1 - \alpha)^{(m-1)^n} \right)(1 - \alpha)^{(m-1)^n} \la\vec{e},\vec{p}^{(n)}  \ra + (1 - \alpha)^{2 (m-1)^n } ||\vec{p}^{(n)}||_2^2 = \\
& =  \left(1 - (1 - \alpha)^{2 (m-1)^n } \right) ||\vec{e}||_2^2 + (1 - \alpha)^{2 (m-1)^n } ||\vec{p}^{(n)}||_2^2 ~.
\end{aligned}
\end{equation}

Since squared norm $||\vec{p}^{(n)}||_2^2$ is bounded by $1$ and $(1 - \alpha)^{2 (m-1)^n }$ converges to 0, we arrive at
\begin{equation}
\lim_{n \to \infty} ||\vec{q}^{(n)} ||_2^2 = ||\vec{e}||_2^2 = \frac{1}{N}~.
\end{equation}
Because each $\vec{q}^{(n)}$ lie inside the probability simplex and $\vec{e}$ is the only vector inside probability simplex with norm equal $1/\sqrt{N}$, we get:
\begin{equation*}
\lim_{n \to \infty} \vec{q}^{(n)} = \vec{e} ~.
\end{equation*}
\end{proof}

Furthermore, we present below a proof of Lemma \ref{sub_multi}, rewritten for convenience.

\begin{lem}
Let $A_{i_1\; \cdots\; i_m}$ be an $m$-stochastic reducible tensor with $m \geq 3$ and let $I$ be a set of indexes defined as above with $\#I = k$. Then the following holds:
\begin{enumerate}
    \item $k \geq N/2$~,
    \item Tensor $A_{i_1\; \cdots\; i_m}$ is reducible with respect to any of its indexes, with the same set $I$ of indexes values,
    \begin{equation}
    \forall r \in \{1\; \cdots\; m\}~,~ \forall~ i_r \in I~,~ i_1,\cdots, i_m \notin I  ~~, A_{i_1\;\cdots\;i_r\;\cdots,i_m} = 0~,
    \end{equation}
    \item The truncation of the tensor $A_{i_1\;\cdots\;i_m}$:
    \begin{equation}\label{sub_stochastic_1_app}
        A'_{i_1\;\cdots\; i_m} = A_{i_1\;\cdots\; i_m} ~,~ i_1, \cdots, i_m \notin I~,
    \end{equation}
    is also an $m$-stochastic tensor. Moreover, for any $m$-stochastic tensor  $A_{i_1\;\cdots\; i_m}$, if there exist a subset $I \subset \{1,\cdots,N\}$ such that a tensor in  \eqref{sub_stochastic_1_app} is also an $m$-stochastic tensor, then $A_{i_1\;\cdots\;i_m}$ is reducible.
\end{enumerate} 
\end{lem}

\begin{proof}
To prove the first statement let us fix the value of indices $i_4, \cdots, i_m \notin I$, and consider the following sum:
\begin{equation}
\sum_{i_1, i_2 \in I,~ i_3\notin I} A_{i_1, \cdots i_m} = \sum_{i_1 \in I, i_3 \notin I} \sum_{i_2} A_{i_1, \cdots i_m} = \sum_{i_1 \in I, i_3 \notin I} 1 = k (N-k)~.
\end{equation}
This expression can be bounded from above by
\begin{equation}
\sum_{i_1, i_2 \in I,~ i_3\notin I} A_{i_1\; \cdots\; i_m} \leq
\sum_{i_1, i_2 \in I} \sum_{i_3} A_{i_1\; \cdots\; i_m} = \sum_{i_1, i_2 \in I} 1 = k^2~.
\end{equation}
Hence we obtain an inequality $k^2 \geq k(N-k)$, implying that $k \geq N/2$.

To prove the second statement notice that
\begin{equation}
\sum_{i_1,\cdots,i_m \notin I} A_{i_1\;\cdots\;i_m} = \sum_{i_2, \cdots, i_m \notin I} \sum_{i_1} A_{i_1\;\cdots\;i_m} = \sum_{i_2, \cdots, i_m \notin I} 1 = (N-k)^{m-1}~.
\end{equation}
Then for each index $i_r$ two equalities hold,

\begin{equation}
\sum_{i_1,\cdots,i_m \notin I}\sum_{i_r} A_{i_1\;\cdots\;i_m} = \sum_{i_1, \cdots, i_m \notin I} 1 = (N-k)^{m-1}~,
\end{equation}
\begin{equation}
\begin{aligned}
& \sum_{i_1,\cdots,i_m \notin I}\sum_{i_r} A_{i_1\;\cdots\;i_m} = \sum_{i_1, \cdots,i_r,\cdots, i_m \notin I} A_{i_1\;\cdots\;i_m}  + \sum_{i_1,\cdots,i_m \notin I}\sum_{i_r \in I} A_{i_1\;\cdots\;i_m} =\\
& \hspace{3.6 cm}= (N-k)^{m-1} + \sum_{i_1,\cdots,i_m \notin I}\sum_{i_r \in I} A_{i_1\;\cdots\;i_m}~.
\end{aligned}
\end{equation}
Hence $\sum_{i_1,\cdots,i_m \notin I}\sum_{i_r \in I} A_{i_1\;\cdots\;i_m} = 0$, which by nonnegativity of the entries $A_{i_1\;\cdots\;i_m}$ implies the second statement.

To show the last statement we need to demonstrate that for each index $i_r$ the sum $\sum_{i_r \notin I} A_{i_1\;\cdots\;i_m} = 1$, if all other indexes are also not in $I$. But by the second statement, if all other indexes are not in $I$, we get
\begin{equation}
\sum_{i_r \notin I} A_{i_1\;\cdots\;i_m} = \sum_{i_r} A_{i_1\;\cdots\;i_m} - \sum_{i_r \in I} A_{i_1\;\cdots\;i_m}  = \sum_{i_r} A_{i_1\;\cdots\;i_m} - 0 = 1~,
\end{equation}
which ends this part of the third statement. To prove the second part of the last statement notice that for each value of indexes $i_2, \cdots i_m \notin I$ the following relations hold,
\begin{equation}
\begin{aligned}
\sum_{i_1} A_{i_1\;\cdots\; i_m} = 1 \text{ and } \sum_{i_1 \notin I} A_{i_1\;\cdots\; i_m} = 1~.
\end{aligned}
\end{equation}
Hence for all $i_1 \in I$, $i_2, \cdots i_m \notin I$ the following entries vanich $A_{i_1\; \cdots\;i_m} = 0$, which ends the proof.
\end{proof}

\subsection{Generalized eingenvectors of quantum multi-stochastic channels}\label{app:quant_fix_point}

In this Appendix we prove Theorem \ref{general_eigenvector_quantum}, using quantum counterparts of techniques applied earlier in the classical case of
$m$-stochastic tensors. To simplify the notation let us denote the maximally mixed state of order $N$ as $\rho_* = \id/N$.
Let us start with the following lemma:

\begin{lem}\label{lem_mstochC}
Let $\Phi_D$ be quantum $m$-stochastic channel. Then for any sequence of density matrices $\{\rho_1, \cdots, \rho_{m-1}\}$, one of which is a maximally mixed state $\rho_*$, the following equality holds:
\begin{equation}\label{eq_lem3.2_0}
\Phi_D[\rho_1\otimes \cdots \otimes \rho_*\otimes\cdots\otimes \rho_{m-1} ] = \rho_* 
\end{equation}
\end{lem}

\begin{proof}
We prove the lemma by induction.
Firstly notice that if $m = 2$, $D$ is just a unital channel and hence $D[\rho_*] =\rho_*$. Otherwise, if $m > 2$, one can define a new $m-1$ stochastic channel as
\begin{equation}\label{eq_lem3.2_1}
\Phi_{D,\rho_{m-1}}[\rho_1 \otimes \cdots\otimes \rho_{m-2}] = \Tr_{1, \cdots m-2}[ D_{\rho_{m-1}}(\rho_1^\top \otimes \cdots \otimes \rho_{m-2}^\top) ]~.
\end{equation}
Where the tensor $D_{\rho_{m-1}}$ has a form,
\begin{equation*}
\Phi_{D,\rho_{m-1}} = \Tr_{m-1}[D ( \mathbb{I}^{\otimes(m-1)} \otimes \rho_{m-1}^\top )]~.
\end{equation*}
If the last state in the tensor product in \eqref{eq_lem3.2_1} is equal to $\rho_*$, we can define an $m-1$ stochastic channel $\Phi_{D,\rho_{m-1}}$ using the first state in an analogous way. By induction, one can assume that the statement \eqref{eq_lem3.2_0} is true for the $m-1$ stochastic channels. Thus we have, 
\begin{equation*}
\Phi_D[\rho_1\otimes \cdots \otimes \rho_*\otimes \cdots \otimes\rho_{m-1} ] = \Phi_{D,\rho_{m-1}}[\rho_1 \otimes \cdots \otimes \rho_*\otimes \cdots\otimes \rho_{m-2}] = \rho_*~,
\end{equation*}
which ends the proof.
\end{proof}

The next theorem is a direct quantum generalization of Theorem \ref{multisto_fixed}.

\begin{thm}\label{D_multisto_fixed}
Let $\sigma^{(0)}$ be a density matrix in the interior of $\Omega_N$ and let $\rho^{(0)}$ be a density matrix at the boundary of $\Omega_N$ such that $\sigma^{(0)} = \alpha \rho_* + (1 -\alpha)\rho^{(0)}$ for some $\alpha \in (0,1]$, Let us define a sequences $\{\sigma^{(n)}\}$, $\{\rho^{(n)}\}$ by $\sigma^{(n+1)} = \Phi_D[\sigma^{(n)} \otimes \cdots \otimes \sigma^{(n)}]$, $\rho^{(n+1)} = \Phi_D[\rho^{(n)} \otimes \cdots\otimes \rho^{(n)}]$ for any $m$-stochastic channel $\Phi_D$. Then if $m > 2$ the sequence $\{\sigma^{(n)}\}$ converges to $\rho_*$.
\end{thm}

\begin{proof}

Let us derive a general formula for $\sigma^{(n)}$ using $\rho^{(n)}$. By explicit calculation of $\sigma^{(1)}$ one obtains:

\begin{equation}
\label{sigma_1}
\begin{aligned}
& \sigma^{(1)} = \Phi_D[\sigma^{(0)}\otimes \cdots \otimes \sigma^{(0)} ] = \Phi_D\left[\left(\alpha \rho_* + (1 - \alpha) \rho^{(0)} \right) \otimes \cdots \otimes\left(\alpha \rho_* + (1 - \alpha) \rho^{(0)} \right) \right] = \\
& ~~~~~ = \sum_{k = 1}^{m-1} \alpha^k (1 - \alpha)^{m-1-k}   \overbrace{ \Phi_D[{\rho_*}^{\otimes k} \otimes {\rho^{(0)}}^{\otimes(m-1-k)} ]+ \cdots }^{\binom{m-1}{k} \text{ terms with $\rho_*$ appearing $k$ times}} = \\
& ~~~~~ = \sum_{k = 1}^{m-2} \alpha^k (1 - \alpha)^{m-1-k} \binom{m-1}{k}  \rho_* + (1 - \alpha)^{m-1} \rho^{(1)} = \\
& ~~~~~ = \left(1 - (1 - \alpha)^{m-1} \right)\rho_* + (1 - \alpha)^{m-1} \rho^{(1)}
\end{aligned}
\end{equation}

The fourth equality above follows from Lemma \ref{lem_mstochC}.
Repeating the above calculations one finds that

\begin{equation}
\label{sigma_n}
\sigma^{(n)} = \left(1 - (1 - \alpha)^{(m-1)^n} \right)\rho_* + (1 - \alpha)^{(m-1)^n} \rho^{(n)}~.
\end{equation}

Now, we may calculate the $|| \cdot ||_2$ norms of $\sigma^{(n)}$ defined by the Hilbert-Schmidt inner product, $\la\rho, \sigma \ra_{HS} = \Tr[\rho^\dagger \sigma]$.
The squared norms of $\sigma^{(n)}$ reads:

\begin{equation}
\begin{aligned}
&||\sigma^{(n)} ||_2^2 = \\
& = \left(1 - (1 - \alpha)^{(m-1)^n} \right)^2 ||\rho_*||_2^2 + 2 \left(1 - (1 - \alpha)^{(m-1)^n} \right)(1 - \alpha)^{(m-1)^n} \la\rho_*,\rho^{(n)}  \ra_{HS} + (1 - \alpha)^{2 (m-1)^n } ||\rho^{(n)}||_2^2 = \\
& =  \left[\left(1 - (1 - \alpha)^{(m-1)^n} \right)^2 + 2(1 - \alpha)^{(m-1)^n} \left(1 - (1 - \alpha)^{(m-1)^n} \right)  \right] \frac{1}{N} + (1 - \alpha)^{2 (m-1)^n } ||\rho^{(n)}||_2^2 = \\
& =  \left(1 - (1 - \alpha)^{2 (m-1)^n } \right) ||{\rho_*}^2||_2^2 + (1 - \alpha)^{2 (m-1)^n } ||\rho^{(n)}||_2^2 ~.
\end{aligned}
\end{equation}
Since $||\rho^{(n)}||_2^2$ is bounded by $1$ and $(1 - \alpha)^{2 (m-1)^n }$ converges to $0$ as $n \to \infty$, we get
\begin{equation}
\lim_{n \to \infty} ||\sigma^{(n)} ||_2^2 = ||\rho_*||_2^2 = \frac{1}{N}~.
\end{equation}
Finally because each $\sigma^{(n)}$ is an element of $\Omega_N$, which is a closed set and $\rho_*$ is the only density matrix inside $\Omega_N$  with norm equal $1/\sqrt{N}$, the limit of interest is
\begin{equation*}
\lim_{n \to \infty} \sigma^{(n)} = \rho_* = \frac{\id}{N} ~.
\end{equation*}
\end{proof}

For further work, we need a quantum counterpart of reducible $m$-stochastic tensor provided in Definition \ref{channel_reducible_1}.
Let us introduce a short-hand notation: $P_{V}$ denotes projection onto a subspace $V \subset \H$ while $P_{V^\perp}$ a projection onto the complementary subspace, orthogonal to $V$. Then for an $m$-stochastic channel, we write the projections of its dynamical matrix in consecutive subsystems,

\begin{equation}
    D|_{V_1, \H, V_3^\perp, \cdots}
    := (P_{V_1}\otimes \id \otimes P_{V_3^\perp}\otimes\cdots)
    D (P_{V_1}\otimes \id \otimes P_{V_3^\perp}\otimes\cdots)^\dagger~.
\end{equation}

The partial trace of the multi-partite matrix $D$ over a subspace $V$ on $k$-th subsystem will be written as

\begin{equation}
    \Tr_{k_V}[D] = \Tr_k[D|_{\H,\cdots,\underset{k\text{-th place}}{V},\cdots, \H}]
\end{equation}

Note that in Definition \ref{channel_reducible_1} we can equivalently use density matrices supported on subspaces $V$ and $V^\perp$ instead of pure states. Hence we use these formulations interchangeably.

\begin{lem}\label{D_sub_multi}
Let $\Phi_D$ be an $m$-stochastic reducible channel with $m \geq 3$ and let $V$ be a proper subspace defined as above with dim$(V) = k$. Then the following properties hold:
\begin{enumerate}
    \item $k \geq N/2$~,
    \item Channel $\Phi_D$ is reducible with respect to any of subsystems, with the same subspace $V$,
    \begin{equation}
    \begin{aligned}
    &\forall t \in \{1, \cdots, m\}~,~ \forall~ |\psi_t\ra \in V,~ |\psi_1\ra,\cdots, |\psi_m\ra \in V^\perp, \\
    & ~~~~~~~~ \Tr[D(|\phi_1\ra\la\phi_1| \otimes \cdots \otimes |\phi_t\ra\la\phi_t| \otimes\cdots\otimes |\phi_m\ra\la\phi_m|)  ] = 0~,
    \end{aligned}
    \end{equation}
    \item The projection of the matrix $D$ of a form $D|_{V^\perp,\cdots,V^\perp}$
    also gives $m$-stochastic channel $\Phi_{D|_{V^\perp,\cdots,V^\perp}}: \Omega_{N-k}^m\to\Omega_{N-k}$. Moreover, for any $m$-stochastic channel  $\Phi_D$, if there exists a subspace $V \subset \H$, such that a $D|_{V^\perp,\cdots,V^\perp}$ defines an $m$-stochastic channel, acting on $\Omega_{N-k}$, then $\Phi$ is reducible.
\end{enumerate} 
\end{lem}

\begin{proof}
Let $\{|e_1 \ra, \cdots |e_k\ra\}$ be an arbitrary basis on $V$ and $\{|f_1\ra, \cdots |  f_{N-k}\ra\}$ be an arbitrary basis on $V^\perp$.
To prove the statement $(i)$ let us fix any vectors $|\psi_4\ra, \cdots, |\psi_m\ra \in V^\perp$, and consider the following expression,
\begin{equation}\label{eq1_dot}
\begin{aligned}
&\Tr_{1_V, 2_V, 3_{V^\perp}}[D|_{|\psi_4\ra\la\psi_4|\otimes\cdots\otimes |\psi_m\ra\la\psi_m|}]=
\Tr_{1_V, 3_{V^\perp}} [\Tr_2[ [D|_{|\psi_4\ra\la\psi_4|\otimes\cdots\otimes |\psi_m\ra\la\psi_m|}]] = \\
& = \sum_{r = 1}^k\sum_{s = 1}^{N-k} \Tr[D|_{|e_r\ra\la e_r|\otimes ~\cdot~\otimes|f_s\ra\la f_s|\otimes|\psi_4\ra\la\psi_4|\otimes\cdots\otimes |\psi_m\ra\la\psi_m|}] = \sum_{r = 1}^k\sum_{s = 1}^{N-k} 1 = k (N-k)~.
\end{aligned}
\end{equation}
But the same sum can also be bounded from above by 
\begin{equation}\label{eq2_dot}
\begin{aligned}
&\Tr_{1_V, 2_V, 3_{V^\perp}}[D|_{|\psi_4\ra\la\psi_4|\otimes\cdots\otimes |\psi_m\ra\la\psi_m|}] \leq
\Tr_{1_V, 2_V}[\Tr_3[D|_{|\psi_4\ra\la\psi_4|\otimes\cdots\otimes |\psi_m\ra\la\psi_m|}]] = \\
& \sum_{r = 1}^k\sum_{s = 1}^{k} \Tr[D|_{|e_r\ra\la e_r|\otimes|e_s\ra\la e_s|\otimes ~\cdot~\otimes|\psi_4\ra\la\psi_4|\otimes\cdots\otimes |\psi_m\ra\la\psi_m|}] = 
\sum_{r = 1}^k\sum_{s = 1}^{k} 1 = k^2~.
\end{aligned}
\end{equation}
Any blank space $\cdot$  in equations \eqref{eq1_dot} and \eqref{eq2_dot} means that the dynamical matrix is truncated analogically as in Lemma \ref{lem_mstochC} with projections inserted onto every subsystem except the second or the third one.
Hence we arrive at the inequality, $k^2 \geq k(N-k)$, so that $k \geq N/2$.

Now we focus on the statement $(ii)$.
In the following calculations, we omit the last steps from \eqref{eq1_dot} and \eqref{eq2_dot}, i.e. the expansion of the traces.
Notice that
\begin{equation}
\Tr_{1_{V^\perp},\cdots,m_{V^\perp}}[D] = \Tr_{2_{V^\perp},\cdots,m_{V^\perp}}[ \Tr_1[D]]  = \sum_{r_2, \cdots, r_m =1}^{N-k} 1 = (N-k)^{m-1}~,
\end{equation}
where in the second to last step we used $m$-stochasticity of $\Phi_D$ combined with Lemma \ref{lem_mstochC}.
Then for each subsystem labelled by $t$ we write in a similar manner,

\begin{equation}
\Tr_{1_{V^{\perp}},\cdots,m_{V^{\perp}}}[\Tr_{t}[ D]] = \sum_{i_1, \cdots, i_m = 1}^{N-k} 1 = (N-k)^{m-1}~.
\end{equation}
But we also have
\begin{equation}
\begin{aligned}
& \Tr_{1_{V^{\perp}},\cdots,m_{V^{\perp}}}[\Tr_{t}[ D]] = \Tr_{1_{V^{\perp}},\cdots,m_{V^{\perp}}}[\Tr_{t_{V^\perp}}[ D]]  + \Tr_{1_{V^{\perp}},\cdots,m_{V^{\perp}}}[\Tr_{t_V}[ D]] =\\
& = (N-k)^{m-1} +Tr_{1_{V^{\perp}},\cdots,t_{V},\cdots,m_{V^{\perp}}}[ D].
\end{aligned}
\end{equation}
Hence one obtains
\begin{equation*}
\begin{aligned}
& 0 = \Tr_{1_{V^{\perp}},\cdots,t_{V},\cdots,m_{V^{\perp}}}[ D] = \sum_{s_1,\cdots,s_m = 1}^{N-1}\sum_{s_t = 1}^k \Tr[D(|f_{s_1}\ra\la f_{s_1}|\otimes\cdots\otimes|e_{s_t}\ra\la e_{s_t}|\otimes \cdots\otimes|f_{s_m}\ra\la f_{s_m}| )] = \\
&  = \sum_{s_1,\cdots,s_m = 1}^{N-1}\sum_{s_t = 1}^k \Tr\left[|f_{s_1}\ra\la f_{s_1}| D[|f_{s_2}\ra\la f_{s_2}|\otimes\cdots\otimes|e_{s_t}\ra\la e_{s_t}|\otimes |f_{s_m}\ra\la f_{s_m}|]\right]~,
\end{aligned}
\end{equation*}
which is a sum of non negative terms, because $D[|f_{s_2}\ra\la f_{s_2}|\otimes\cdots\otimes |f_{s_m}\ra\la f_{s_m}|]$ is semipositive define density matrix, and so is $|f_{s_1}\ra\la f_{s_1}|$. Therefore, each of those terms is equal to $0$, which by multi-linearity of the channel $\Phi_D$ and by the fact that bases $|e_i\ra$, $|f_i\ra$ can be chosen arbitrarily in each subsystem, proves the statement $(ii)$.

In the statement $(iii)$ the semi positivity of the channel $D|_{V^\perp,\cdots V^\perp}$ follows immediately from semi positivity of  $D$. Thus we are left to show  the $m$-stochasticity of the channel $\Phi{D|_{V^\perp,\cdots,V^\perp}}$ that for any subsystem labeled by $t$ and any density matrices $\rho_1,\cdots,\rho_m$ supported in $V^\perp$ one has
\begin{equation*}
\Tr_{V^\perp}\left[\Tr_{1,\cdots,t-1,t+1\cdots,m}[D|_{V^\perp,\cdots,V^\perp}(\rho_1\otimes\cdots\otimes\underset{t\text{-th place}}{\id_{V^\perp}}\otimes\cdots\otimes\rho_m)]\right] \overset{?}{=} 1
\end{equation*}

Since all density matrices $\rho_l$ are supported in $V^\perp$, we get the following chain of equalities:
\begin{equation}
\begin{aligned}
&\Tr_{V^\perp}\left[ \Tr[D|_{V^\perp,\cdots,V^\perp}(\rho_1\otimes\cdots\otimes\underset{t\text{-th place}}{\id_{V^\perp}}\otimes\cdots\otimes\rho_m)]\right]  = \Tr_{V^\perp}\left[\Tr[D(\rho_1\otimes\cdots\otimes\underset{t\text{-th place}}{\id_{V^\perp}}\otimes\cdots\otimes\rho_m)]\right] = \\
& = \Tr[D(\rho_1\otimes\cdots\otimes\underset{t\text{-th place}}{\id}\otimes\cdots\otimes\rho_m)] - \Tr_V[D(\rho_1\otimes\cdots\otimes\underset{t\text{-th place}}{\id_V}\otimes\cdots\otimes\rho_m)] = 1-0 = 1,
\end{aligned}
\end{equation}
where the second to last step follows from the multi-linearity and the second statement. 

To prove the second part of the statement $(iii)$ let us choose an arbitrary set of states $|\psi_2\ra,\cdots,|\psi_m\ra \in V^\perp$, so that
\begin{equation}
\Tr [ D|_{V^\perp,\cdots,V^\perp}[|\psi_2\ra\la\psi_2|\otimes\cdots\otimes|\psi_m\ra\la\psi_m|]] = \Tr_{V^\perp}[D[|\psi_2\ra\la\psi_2|\otimes\cdots\otimes|\psi_m\ra\la\psi_m|]] =  1~,
\end{equation}
and
\begin{equation}
\Tr [D[|\psi_2\ra\la\psi_2|\otimes\cdots\otimes|\psi_m\ra\la\psi_m|]] =  1.
\end{equation}
Hence for all $|\psi_1\ra \in V$ one has
\begin{equation*}
\Tr [D(|\psi_1\ra\la\psi_1|\otimes|\psi_2\ra\la\psi_2|\otimes\cdots\otimes|\psi_m\ra\la\psi_m|)] = \la\psi_1|D[|\psi_2\ra\la\psi_2|\otimes\cdots\otimes|\psi_m\ra\la\psi_m|]    |\psi_1\ra = 0 ~,
\end{equation*}
which ends the proof.
\end{proof}

Now we are ready to present a proof of Theorem \ref{general_eigenvector_quantum}, which we invoke here for completeness.

\begin{thm}
Let $\Phi_D$ be an $m$-stochastic channel acting on the $N$-dimensional states and $m \geq 3$. Then each subspace $V\subset \H$ such that
\begin{equation*}
\forall~ |\phi_1\ra \in V~,~ |\psi_2\ra,\cdots, |\phi_m\ra \in V^\perp ~~\Tr[D(|\phi_1\ra\la\phi_1| \otimes |\phi_2\ra\la\phi_2| \otimes\cdots\otimes |\phi_m\ra\la\phi_m|)  ] = 0 ~,
\end{equation*}
corresponds to a single density matrix being a generalized eigenvector of $\Phi_D$ and vice versa. Moreover, each such eigenvector is a maximally mixed state on the subspace $V^{\perp}$ and the corresponding eigenvalue is $1$.
\end{thm}

\begin{proof}
We demonstrate first that each eigenvector $\rho$ corresponds to a subspace $V$ described in the theorem and then show that each subspace $V$ corresponds to an eigenvector $\rho$. Along the way, we also prove the second statement of the theorem. 

Let $\rho$ be an eigenvector of $D$ to the eigenvalue $\lambda$. If all eigenvalues of $\rho$ are greater than zero, then by  Theorem \ref{D_multisto_fixed} one has $\rho = \rho_*$ and $\lambda = 1$, since $\rho$ belongs to the interior of $\Omega_N$ so the theorem is satisfied with $V = 0$.
Otherwise, there exist a subspace $V$ spanned by eigenvectors $|e_i\ra$ corresponding to zero eigenvalues of $\rho$, hence $P_{V^\perp} \rho P_{V^\perp} = \rho$. For each such $|e_i\ra$ one has
\begin{equation}\label{eqD_thm_proof1}
\begin{aligned}
0 & = \lambda \la e_{i_1}|\rho|e_{i_1}\ra=  \la e_{i_1} \Phi_D[\rho\otimes\cdots\otimes\rho] e_{i_1}\ra =\\
& =\sum_{i_2, \cdots i_m = 1}^{N-k} \lambda_{i_2}\cdots\lambda_{i_m}    \Tr[D(| e_{i_1}\ra\la e_{i_1}|\otimes |f_{i_2}\ra\la f_{i_2}|\otimes \cdots \otimes|f_{i_m}\ra\la f_{i_m}| )] \geq \\
& \geq \sum_{i_2, \cdots i_m = 1}^{N-k} \lambda_{\text{min}}^{m-2} \Tr[D(| e_{i_1}\ra\la e_{i_1}|\otimes |f_{i_2}\ra\la f_{i_2}|\otimes \cdots \otimes|f_{i_m}\ra\la f_{i_m}| )] = \\
& = \lambda_{\text{min}}^{m-2} \Tr_{1_V , 2_{V^\perp},\cdots,m_{V^\perp}}[D]~,
\end{aligned}
\end{equation}
where each $\lambda_{i_r}$ and $|f_{i_r}\ra$ are a $i_r$-th positive eigenvalues and the corresponding eigenvector of $\rho$ and $\lambda_{\text{min}}$ is the smallest positive eigenvalue of $\rho$. Because $D$ is semipositive and its trace on $V\otimes V^\perp \otimes \cdots\otimes V^\perp \subset \H\otimes\cdots\otimes\H$ is equal to zero, the dynamical matrix $D$ must be identically zero on this subspace, which proves the thesis in one direction.

By Lemma \ref{D_sub_multi} the dynamical matrix $D$ can be truncated to $D_{V^\perp,\cdots,V^\perp}$, which defines an $m$-stochastic channel. On the subspace $V^\perp$, $\rho$ is an eigenvector of $\Phi_{D_{V^\perp,\cdots,V^\perp}}$ with no eigenvalues equal to $0$. Therefore, by Theorem \ref{D_multisto_fixed} applied to the channel $\Phi_{D_{V^\perp,\cdots,V^\perp}}$, the state $\rho|_{V^\perp}$ must be a totally mixed state on subspace $V^\perp$ and the generalized eigenvalue corresponding to $\rho$ is equal $1$.

To prove the implication in the opposite direction we use the third point of Lemma \ref{D_sub_multi}. Because a channel $\Phi_{D_{V^\perp,\cdots,V^\perp}}$ is $m$-stochastic, by Lemma \ref{lem_mstochC} it also has a generalized eigenvector $\rho := \rho_*$ on the subspace $V^\perp$. 
Moreover, since $\Phi_D$ is reducible we know the second point of Lemma \ref{D_sub_multi} also holds, hence
\begin{equation}
\begin{aligned}
& \Phi_D[\rho\otimes\cdots\otimes\rho] = P_{V^\perp}\Phi_D[\rho\otimes\cdots\otimes\rho] P_{V^\perp} = \Tr_{2,\cdots, m}[D (P_{V^\perp}\otimes\rho^\top\otimes\cdots\otimes\rho^\top) ] = \\
& = \Tr_{2,\cdots,m }[D|_{V^{\perp}\cdots V^\perp}(P_{V^\perp}\otimes\rho^\top\otimes\cdots\otimes\rho^\top) ]  =  \Phi_{D|_{V^{\perp}\cdots V^\perp}}[\rho\otimes\cdots\otimes\rho] = \rho~,
\end{aligned}
\end{equation}

so $\rho$ is also an eigenvector of $D$.
\end{proof}

\subsection{Coherification of a tristochastic tensor of dimension $2$}\label{app:qutib_gen_coch}

In this Appendix we propose a way to construct a coherification of an arbitrary two-dimensional tristochastic tensor represented by a cube of size 2 
\begin{equation}
\label{eqAx}
A(x) = \left(\begin{matrix}
x & 1-x\\
1-x & x\\
\end{matrix}\right.
\left|\begin{matrix}
1-x & x\\
x & 1-x\\
\end{matrix}\right)~,
\end{equation}
with $x \in [0,1]$.
Notice, that the fixed diagonal values of the dynamical matrix $D$ imply that all of the Kraus operators describing the operation $\Phi_x$ have the form:
\begin{equation}\label{2x2x2_Kraus}
K_i = \left(\begin{matrix}
a_i & b_i & c_i & d_i \\
e_i & f_i & g_i & h_i \\
\end{matrix}\right)~,
\end{equation}
with $||a||^2 = ||f||^2 = ||g||^2 = ||d||^2 = x$ and $||e||^2 = ||b||^2 = ||c||^2 = ||h||^2 = 1-x$.

Moreover, the condition $\sum_i K_i^\dagger K_i = \id$ leads to the following constrains
\begin{equation}\label{2x2x2_orto}
 \left(\begin{matrix}
||a||^2 + ||e||^2 & \la a|b \ra + \la e|f \ra & \la a|c \ra + \la e|g \ra & \la a|d \ra + \la e|h \ra \\
\la b|a \ra + \la f|e \ra & ||b||^2 + ||f||^2 & \la b|c \ra + \la f|g \ra & \la b|d \ra + \la f|h \ra \\
\la c|a \ra + \la g|e \ra & \la c|b \ra + \la g|f \ra & ||c||^2 + ||g||^2 & \la c|d \ra + \la g|h \ra \\
\la d|a \ra + \la h|e \ra & \la d|b \ra + \la h|f \ra & \la d|c \ra + \la h|g \ra & ||d||^2 + ||h||^2 \\
\end{matrix}\right) = 
\left(\begin{matrix}
1 & 0 & 0 & 0 \\
0 & 1 & 0 & 0 \\
0 & 0 & 1 & 0 \\
0 & 0 & 0 & 1 \\
\end{matrix}\right)
\end{equation}

These conditions may be satisfied using only two Kraus operators. 
Assume that the following pairs of vectors are orthogonal, $|a\ra \perp |d\ra$, $|e\ra \perp |h\ra$, $|b\ra \perp |c\ra$, $|f\ra \perp |g\ra$, so the only relations left to satisfy are:
\begin{equation*}
\begin{aligned}
& \la a|b \ra = - \la e|f \ra~, &&  \la a|c \ra = - \la e|g \ra~, \\
& \la d|b \ra = - \la h|f \ra~, &&  \la d|c \ra = - \la h|g \ra~. \\
\end{aligned}
\end{equation*}

Then we can define four ortonormal bases $\{|\tilde{a}\ra, |\tilde{d}\ra\}$, $\{|\tilde{e}\ra, |\tilde{h}\ra\}$, $\{|\tilde{b}\ra, |\tilde{c}\ra\}$, $\{|\tilde{f}\ra, |\tilde{g}\ra\}$, by rescaling the above pairs of vectors. Let $U$ be a unitary matrix changing the basis $\{|\tilde{a}\ra, |\tilde{d}\ra\}$ into a basis $\{|\tilde{e}\ra, |\tilde{h}\ra\}$. Then if the bases $\{|\tilde{b}\ra, |\tilde{c}\ra\}$ and $\{|\tilde{f}\ra, |\tilde{g}\ra\}$ are connected by $-U$, the condition  \eqref{2x2x2_orto} is satisfied. Such dependency between bases $\{|\tilde{a}\ra, |\tilde{d}\ra\}$, $\{|\tilde{e}\ra, |\tilde{h}\ra\}$, $\{|\tilde{b}\ra, |\tilde{c}\ra\}$, $\{|\tilde{f}\ra, |\tilde{g}\ra\}$ is our second assumption.

Using these assumptions one can calculate the norm $2$ coherence \cite{Korzekwa_coherifying}:

\begin{equation}
\begin{aligned}
&C_2(\Phi_x) =  \sum_{ijkl}|(\rho_\Phi)^{i\;k\;}_{\;j\;l}|^2 - \sum_{ik}|(\rho_\Phi)^{i\;k\;}_{\;i\;k}|^2 = \frac{1}{4^2}  \sum_{ijkl}|(J_\Phi)^{i\;k\;}_{\;j\;l}|^2  - \frac{1}{4^2} \sum_{ik}|(J_\Phi)^{i\;k\;}_{\;i\;k}|^2=\\
& = \frac{1}{16} 2\Big\{ |\la a| b\ra|^2 + |\la a| c\ra|^2 + |\la a| e\ra|^2 + |\la a| f\ra|^2 + |\la a| g\ra|^2 + |\la a| h\ra|^2 + \\
& ~~~~~~~+ |\la d| b\ra|^2 + |\la d| c\ra|^2 + |\la d| e\ra|^2 + |\la d| f\ra|^2 + |\la d| g\ra|^2 + |\la d| h\ra|^2 + \\
& ~~~~~~~+ |\la b|e \ra|^2 + |\la b|f \ra|^2 + |\la b|g \ra|^2 + |\la b|h \ra|^2 + |\la c|e \ra|^2 + |\la c|f \ra|^2 + |\la c|g \ra|^2 + |\la c|h \ra|^2  + \\
& ~~~~~~~+ |\la e|f \ra|^2 + |\la e|g \ra|^2 + |\la h|f \ra|^2 + |\la h|g \ra|^2 \Big\}~.
\end{aligned}
\end{equation}
Note, that in the first two lines we have rescaled projections of consecutive vectors onto a basis $\{|\tilde{a}\ra, |\tilde{d}\ra\}$, 
in the second to last line we have rescaled projections onto basis vectors $|\tilde{b}\ra$, and $|\tilde{c}\ra$, and in the last line rescaled projections onto basis vectors $|\tilde{e}\ra$, and $|\tilde{h}\ra$. Hence the formula for the coherence of the map resulting from the classical tensor $A(x)$ defined in eq. \eqref{eqAx} reads
\begin{equation}
C_2(\Phi_x) = \frac{1}{8} \left\{x(4(1-x)+ 2x)+ (1-x)(2(1-x)+2x) + 2(1-x) x \right\} = \frac{1}{4}(1 + 2x - 2x^2) ~.
\end{equation}

Notice that there cannot exist a channel $\Omega_2 \otimes \Omega_2 \to \Omega_2$ described by only one Kraus operator, because it would imply that such channel is a unitary between vector spaces of different dimensions. Therefore, described coherifications use a minimal number of Kraus operators. The squared norms of each Kraus operator rescaled by $1/4$ are eigenvalues of $\rho_\Phi$ and each of them is equal to:

\begin{equation*}
\begin{aligned}
& \frac{1}{4} ||K_i ||^2 = \frac{1}{16} \left[|a_i|^2 + |b_i|^2 + |c_i|^2 + |d_i|^2 + |e_i|^2 + |f_i|^2 + |g_i|^2 + |h_i|^2\right] =  \\
& = \frac{1}{4}\left[|\la i| a \ra|^2 + |\la i| d \ra|^2  + |\la i| e \ra|^2 + |\la i| h \ra|^2+|\la i| b \ra|^2 + |\la i| c \ra|^2+|\la i| f \ra|^2 + |\la i| g \ra|^2\right] = \\
& = \frac{1}{4} [x + (1-x) + x + (1-x)] = \frac{1}{2}~,
\end{aligned}
\end{equation*}
where we used the fact that $\{|a\ra$, $|d\ra\}$, $\{|e\ra$, $|h\ra\}$, $\cdots$ are rescaled basis vectors. Hence the entropic coherence \cite{Korzekwa_coherifying} of the map $\Phi_x$ is equal to

\begin{equation*}
C_S(\Phi_x) = S(\text{diag}(\rho_{\Phi})) - S(\rho_\Phi) = -x\ln\left(\frac{x}{4}\right) - (1-x) \ln\left(\frac{1-x}{4}\right) + \ln(2)~,
\end{equation*}
which for $x = 1/2$ goves the maximal value $C_S^{max} = 4 \ln(2)$.

\subsection{Entangling power and gate typicality of cocherifications of permutation tensor of dimension $2$}\label{app:ep_gt}

In this Section we study the properties of the unitary matrix $U$: $\H \otimes \H \to \H \otimes \H$ which we used to construct a coherification of permutation tensor:
\begin{equation*}
D[\rho_1, \rho_2] = \Tr_2[U(\rho_1 \otimes \rho_2)U^\dagger] 
\end{equation*}
We focus on two quantities describing the action of two-qubit channels: entangling power and gate typicality. 
Following \cite{ep_Zanardi} we define entangling power of a bipartitie gate $U$:
\begin{equation}
    e_p(U) = \frac{N+1}{N-1}\overline{\mathcal{E}(U|\psi_A\ra|\psi_B\ra)}^{\psi_A \psi_B}~.
\end{equation}

As a normalized linear entropy $\mathcal{E}(|\psi\ra) = 1 - \Tr_A[\Tr_B(|\psi\ra\la\psi|)^2]$ of output averaged over all possible input product states with Haar measure.
It can be calculated by the entanglement entropies of gates $U$ and $US$, where $S$ is the \textit{SWAP} operator: $S|\psi_A\ra|\psi_B\ra = |\psi_B\ra|\psi_A\ra$ \cite{ep_Zanardi},

\begin{equation*}
e_p(U) = \frac{1}{E(S)}\left(E(U) + E(SU) - E(S) \right)~.
\end{equation*}

Entangling power $e_p(U) \in [0,1]$ tells us how entangled is on average the output of the gate $U$ if the input is separable. If one performs the partial trace over a single subsystem, as it is required to obtain the convolution channel $\Phi_D$, then the larger the entangling power, the more mixed the output of the channel $\Phi_D$ is.

The complementary quantity to entangling power is gate typicality $g_t(U)$ \cite{JMZL_entanglement_measures}

\begin{equation*}
g_t(U) = \frac{1}{2 E(S)}\left(E(U) - E(US) + E(S) \right)~.
\end{equation*}

The gate typicality, $g_t(U) \in [0,1]$, tells us how much on average the systems are "interchanged". For example for local rotations $g_t(U_A \otimes U_B) = 0$ and for swap $g_t(S) = 1$. Hence, if one performs a partial trace over the second subsystem after $U$, to obtain convolution channel $\Phi_D$, the gate typicality tells us, which of the input state is more relevant to the output of $D$. 

For the convolution of two qubits $U = U_4$  given in \eqref{two_q_channel}, one obtains by direct calculation:

\begin{equation}
e_p(U_4) = \frac{2}{3} ~,~~~ g_t(U_4) = \frac{1}{6}\left(3 - \cos\theta \right).
\end{equation}

Therefore the entangling power of 2-qubit channels $U_4$ is maximal for any values of parameters $\alpha$, $\theta$, $\phi$. The gate typicality depends only on the parameter $\theta$ and attains all possible values $\left[\frac{1}{3},\frac{2}{3}\right]$ for a given entangling power. 

\begin{figure}[h]
    \centering
        \includegraphics[height=3in]{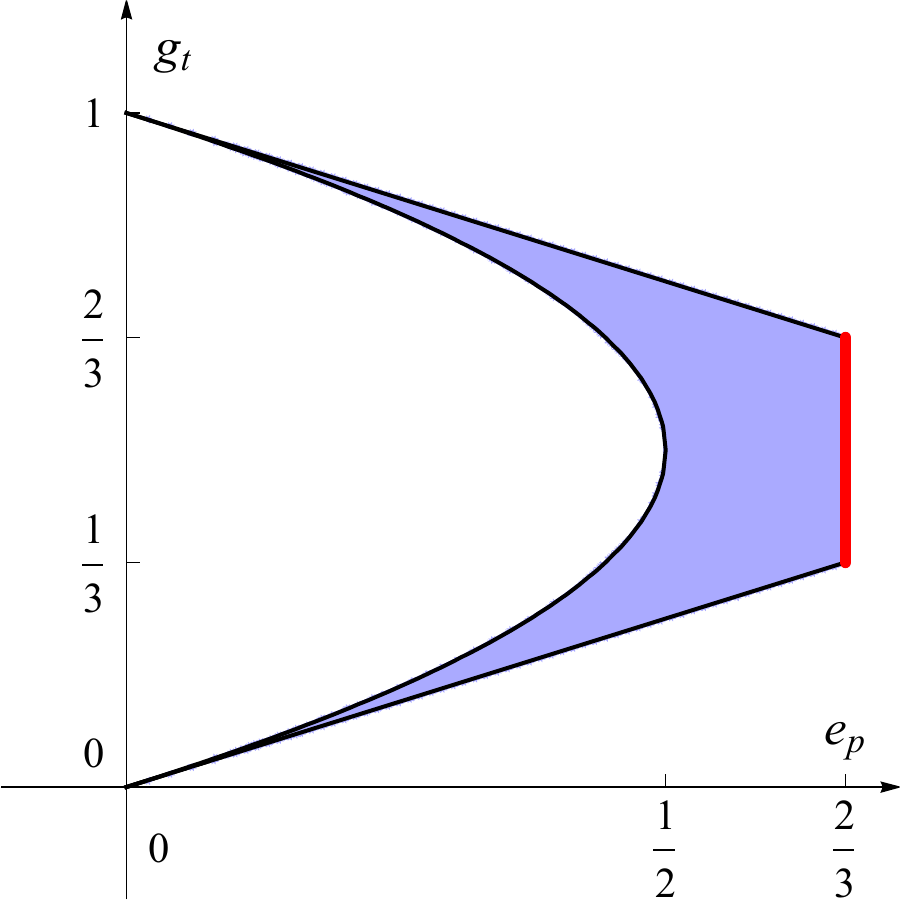}
     \centering
     \caption{\label{fig:ep_gf} Projection of the set of two-qubit unitary gates $U \in U(4)$ into the plane: entangling power $e_p$ versus gate typicality $g_t$,
     limited by black and red borders. The gates $U_4$ defined in \eqref{two_q_channel}, corresponding to quantum convolution, occupy the red bold border with the highest possible entanglement power $e_p = 2/3$ for two-qubit gates.}
\end{figure}

\section{Explicit forms of convolution of quantum states}\label{app:explicit_form}

The aim of this appendix is to explicitly provide the formulas for convolution between two quantum states.
First let us show the dynamical matrix of the optimal coherification of bit permutation tensor \eqref{bigchoi}, using a convenient parametrization from \eqref{two_q_channel_krauss}:

\begin{equation}
\label{bigchoi_2}
D_T = \left(
\begin{array}{cccccccc}
 1 & 0 & 0 & 0 & 0 & e^{-i \alpha} \cos{\frac{\theta}{2}} & e^{-i \alpha} \sin{\frac{\theta}{2}} & 0 \\
 0 & 0 & 0 & 0 & 0 & 0 & 0 & 0 \\
 0 & 0 & 0 & 0 & 0 & 0 & 0 & 0 \\
 0 & 0 & 0 & 1  & 0 & e^{-i \alpha}e^{-i \phi} \sin{\frac{\theta}{2}} & -e^{-i \alpha}e^{-i \phi} \cos{\frac{\theta}{2}} & 0 \\
 0 & 0 & 0 & 0 & 0 & 0 & 0 & 0 \\
 e^{i \alpha} \cos{\frac{\theta}{2}} & 0 & 0 & e^{i \alpha}e^{i \phi} \sin{\frac{\theta}{2}} & 0 & 1 & 0 & 0 \\
 e^{i \alpha} \sin{\frac{\theta}{2}} & 0 & 0 & -e^{i \alpha}e^{i \phi} \cos{\frac{\theta}{2}} & 0 & 0 & 1 & 0 \\
 0 & 0 & 0 & 0 & 0 & 0 & 0 & 0 \\
\end{array}
\right)~.
\end{equation}

As we argued in Section \ref{sec:Qconvolution}, this channel is easier to study as a composition of unitary channel $U_4(\alpha, \theta,\phi)$ described in \eqref{two_q_channel}, followed by a partial trace over the second subsystem. 
Within this class of unitary maps, at least a couple of well-known bipartite gates are hidden. 

For instance, for $\alpha = \phi = 0$, $\theta = \frac{\pi}{2}$,  the gate typicality is equal $g_t = \frac{1}{3}$ and

\begin{equation*}
U_4\left(0,0,\frac{\pi}{2}\right) = 
DCNOT,
\end{equation*}

whereas for $\alpha = 0$, $\phi = \pi$, $\theta = 0$, the  gate typicality is equal $g_t = \frac{2}{3}$ and the resulting gate is

\begin{equation*}
U_4\left(0,\pi,0\right) = 
 S \cdot CNOT = CNOT \cdot DCNOT
\end{equation*}

To obtain the final expression for the convolution it is sufficient to plug matrix $D$ from \eqref{bigchoi_2} to formula \eqref{quantum_convolution}.

For larger systems the convolution, in the sense of optimal coherification (Definition \ref{def_coch}) of permutation tensor, is also easier to study and implement in the form of a unitary channel followed by a partial trace. This form was already presented in the Theorem \ref{lem_channel_generation}, together with an example of the qutrit channel.


\vspace{1 cm}
\bibliographystyle{plain}

\end{document}